\documentclass[sn-mathphys-num]{sn-jnl}

\usepackage{graphicx}%
\usepackage{multirow}%
\usepackage{amsmath,amssymb,amsfonts}%
\usepackage{amsthm}%
\usepackage{mathrsfs}%
\usepackage[title]{appendix}%
\usepackage{xcolor}%
\usepackage{textcomp}%
\usepackage{manyfoot}%
\usepackage{booktabs}%
\usepackage{algorithm}%
\usepackage{algorithmicx}%
\usepackage{algpseudocode}%
\usepackage{listings}%

\theoremstyle{thmstyleone}%
\newtheorem{theorem}{Theorem}
\newtheorem{proposition}[theorem]{Proposition}%
\newtheorem{corollary}[theorem]{Corollary}%

\theoremstyle{thmstyletwo}%
\newtheorem{example}{Example}%

\theoremstyle{thmstylethree}%
\newtheorem{definition}{Definition}%

\usepackage[draft]{fixme}
\usepackage{amsmath}
\usepackage{amssymb}
\usepackage{xcolor}
\usepackage[english]{babel}
\usepackage[latin1]{inputenc}
\usepackage{dsfont}
\usepackage{tikz}
\usetikzlibrary{backgrounds}
\usetikzlibrary{patterns}
\usepackage{tkz-berge}
\usepackage{enumitem}
\usepackage{multicol}
\usepackage{paracol}

\usepackage{cleveref}
\usepackage{stmaryrd}

\usepackage{localdefs}

\newif\ifcategories

\categoriesfalse

\begin{document}

\title[Deontic Action Logics. A Modular Algebraic Perspective]{\centering
  Deontic Action Logics: \\ A Modular Algebraic Perspective}


\author*[1,2]{\fnm{Carlos} \sur{Areces}}\email{carlos.areces@unc.edu.ar}
\author*[1,3]{\fnm{Valentin} \sur{Cassano}}\email{valentin@dc.exa.unrc.edu.ar}
\author*[1,3]{\fnm{Pablo} \sur{Castro}}\email{pcastro@dc.exa.unrc.edu.ar}
\author*[1,2]{\fnm{Raul} \sur{Fervari}}\email{rfervari@unc.edu.ar}

\affil[1]{
  \orgname{Consejo Nacional de Investigaciones Cient\'ificas y T\'ecnicas},
  \orgaddress{
    \country{Argentina}}}

\affil[2]{
  \orgname{Universidad Nacional de C\'ordoba},
  \orgaddress{
    \country{Argentina}}}
  
\affil[3]{
  \orgname{Universidad Nacional de R\'io Cuarto},
  \orgaddress{
    \country{Argentina}}}


\abstract{
	In a seminal work, K. Segerberg introduced a deontic logic called \DAL to investigate normative reasoning over actions. \DAL marked the beginning of a new area of research in Deontic Logic by shifting the focus from deontic operators on propositions to deontic operators on actions.
	In this work, we revisit \DAL and provide a complete algebraization for it. In our algebraization we introduce deontic action algebras --algebraic structures consisting of a Boolean algebra for interpreting actions, a Boolean algebra for interpreting formulas, and two mappings from one Boolean algebra to the other interpreting the deontic concepts of permission and prohibition.
	We elaborate on how the framework underpinning deontic action algebras enables the derivation of different deontic action logics by removing or imposing additional conditions over either of the Boolean algebras. We leverage this flexibility to demonstrate how we can capture in this framework several logics in the \DAL family.
	Furthermore, we introduce four variations of \DAL by:
	(a) enriching the algebra of formulas with propositions on states,
	(b) adopting a Heyting algebra for state propositions,
	(c) adopting a Heyting algebra for actions, and
	(d) adopting Heyting algebras for both.
	We illustrate these new deontic action logics with examples and establish their algebraic completeness.
}




\maketitle



\section{Introduction}\label{section:introduction}

The logical laws of normative reasoning have attracted the attention of philosophers, lawyers, logicians, and computer scientists since the beginning of their disciplines~\cite{Gabbay:2013}.
The first modern deontic systems go back to the pioneer
works of von Wright~\cite{vonWright:1951},
Kalinowski~\cite{Kalinowski53}, and Becker~\cite{Becker52}.
These first systems, that were conceived on the view that prescriptions apply to actions, were swiftly overtaken by modal approaches where prescriptions applied to propositions~\cite{blac:moda00}.
Indeed, it can be said that the influence of Modal Logic marked a crossroads at the very beginning of Deontic Logic.
One path leads to the so-called \emph{ought-to-be} deontic systems, where prescriptions apply to propositions; the other path leads to the so-called \emph{ought-to-do} deontic systems, where prescriptions apply to actions.



The most easily recognized representative of the ought-to-be systems is \emph{Standard Deontic Logic} (\SDL)~\cite{Aqvist:2002}.
\SDL extends the normal modal system \textsf{K} with an axiom for \emph{seriality}.
In \SDL, the ``box'' modality, written $\mathsf{O}\varphi$, is informally read as ``$\varphi$ is obligatory''.
Resorting to this modality, we can formally discuss the implications of how to understand obligations.
For instance, if $w$ formalizes the proposition ``John is behind the steering wheel'', and $i$ formalizes the proposition ``John is intoxicated'', we can explore which one of $w \to \mathsf{O}(\lnot i)$, $\mathsf{O}(w \to \lnot i)$, or $\mathsf{O}(w \land \lnot i)$, better captures the proposition ``it is obligatory that John is not intoxicated while being behind the steering wheel''.
This kind of propositions are discussed at length in the literature on Deontic Logic and riddled with challenges and paradoxes~\cite{Aqvist:2002}.

The counterpart to \SDL for ought-to-do systems is, arguably, the \emph{Deontic Action Logic} (\DAL) proposed by Segerberg in~\cite{Segerberg1982}.
The language of \DAL distinguishes between actions and formulas. Actions are built up from basic action names using action operators.
Then, deontic connectives $\perm$ of permission and $\forb$ of prohibition are applied to actions to yield formulas that can be combined using logical connectives.
Let us illustrate this by means of a simple example.  

\medskip

\begin{example}
Let $\mathsf{driving}$ and $\mathsf{drinking}$ be basic action symbols indicating the \emph{acts} of driving, and drinking, respectively.
Moreover, let $\mathsf{driving} \sqcap \mathsf{drinking}$ be understood as the parallel composition of the actions $\mathsf{driving}$ and $\mathsf{drinking}$.
Then, $\forb(\mathsf{driving} \sqcap \mathsf{drinking})$ is a formula intuitively taken to assert
that drinking while driving is forbidden. 
\end{example}

\medskip

The work of Segerberg in \cite{Segerberg1982} shifted the focus from deontic operators on propositions to deontic operators on actions. 
Moreover, it gave rise to what nowadays can be construed as a family of deontic action logics~\cite{Maibaum87,Segerberg1982,Meyer:1994,Castro:2009,BroersenThesis,Trypuz:2010,Trypuz15,Prisacariu:2012}.
In addition to being interesting from a purely logical perspective, the logics in this family are good canditates as formalisms to describe the behavior of real world systems.
For instance, in~\cite{Castro:2009,DemasiCRMA15}, a variant of \DAL is used to reason about fault-tolerance. 
Therein, actions formalize changes of state in a system, permitted actions indicate the normal behavior of the system, while forbidden actions are used to model the faulty behavior of the system.
This classification of actions paves the way to reasoning about how to react in response to faults.
In this setting, we can, for instance, understand a formula such as $\forb(\mathsf{read} \sqcap \mathsf{write})$ as prescribing the behavior of a system by indicating that it is forbidden to simultaneously read and write from a memory location, as this could lead to the system being in an inconsistent state.
The falsehood of this formula in a particular scenario serves as an indication of faulty behavior, indicating the necessity of fault-tolerant mechanisms to safeguard the normal operation of the system.



\medskip\noindent
\textbf{Proposal.}
We take \DAL as the starting point to investigate the construction of deontic action logics in an algebraic framework.
To this end, we build on an earlier work where we develop an abstract view of \DAL resorting to algebraic structures called \emph{deontic action algebras}~\cite{CCFA:2021}.
In brief, a deontic action algebra consists of a Boolean algebra for interpreting actions, a Boolean algebra for interpreting formulas, and two mappings from one Boolean algebra to the other interpreting the deontic concepts of permission and prohibition.
We use deontic action algebras to formally interpret the two tier structure of the language of \DAL.
An interesting feature of this algebraic treatment of \DAL is that it is modular, giving rise to natural variations.
We explore this modularity by elaborating of how to capture various logics in the \DAL family.


\medskip\noindent
\textbf{Contributions.} 
This paper continues and extends our work in \cite{CCFA:2021}. First, we revisit the algebraic framework of deontic action algebras, provide detailed soundness and completeness proofs, and add motivating examples. 
Second, we present a series of deontic action logics which exploit the modularity of the algebraic framework.
We begin by enriching the algebra for formulas with propositions to describe \emph{states of affairs}.  In this way, we can express both properties of actions, and propositions (e.g., pre- and post-conditions) about the states in which this actions take place. 
Then, we discuss the result of replacing the Boolean algebra interpreting formulas by a Heyting algebra.
The resulting extension of \DAL can deal with scenarios in which laws like the excluded middle or contraposition might be rejected.
This could be the case, for example, in normative systems in which evidence is required in order to accept some assertions as true.
In turn,  we consider using a Heyting algebra to interpret actions. We argue that this would be useful, e.g., when actions are associated to constructions witnessing their \emph{realizability}.
This admits a direct analogy with the standard interpretation of intuitionistic logic in which the concept of truth is associated to that of proof.
Clearly, we can do both at the same time, obtaining a fully intuitionistic deontic action logic were both actions and formulas are interpreted using Heyting algebras. 
In all cases, we obtain axiom systems that are sound and complete for the corresponding classes of deontic action algebras.

Algebraic logic has been shown useful for analyzing theoretical properties of logics and for investigating relations between different logics~\cite{ras1963,Andreka1991-ANDAL-2}. We believe that the case of \DAL presents a particularly simple and elegant instance of this framework. 


\medskip\noindent
\textbf{Structure of the Paper.} \Cref{section:dal} introduces basic definitions, and Segerberg's deontic action logic \DAL.  \Cref{sec:algebraic-char} presents the basic algebraic framework,
and prove soundness and completeness for \DAL
using standard algebraic tools. Variations of \DAL are investigated in~\Cref{section:new:dals}.
\Cref{section:conclusion} offers some final remarks and discusses
future work.

\section{Deontic Action Logic}\label{section:dal}

In this section, we cover the language, semantics, and axiomatization of Segerberg's deontic action logic \DAL~\cite{Segerberg1982}.  We also state a soundness and completeness result for future reference.

\paragraph{Language.}\label{section:dal:syntax}

The language of {\DAL} consists of \emph{actions} and \emph{formulas}.
Actions, indicated $\alpha$, $\beta$, $\gamma$, \dots, are built on a countable set $\bact = \set{\mathsf{a}_i}{i\in \Nat_0}$ of basic action symbols according to the following grammar:
\begin{align}
		\alpha,\beta & ::=
				\mathsf{a}_i
			\mid
				{\alpha \sqcup \beta}
			\mid
				{\alpha \sqcap \beta}
			\mid
				{\bar{\alpha}}
			\mid
				\mathsf{0}
			\mid
				\mathsf{1}. \label{definition:actions}
\end{align}%
We use $\act$ to indicate the set of all actions.
Formulas of \DAL, indicated $\varphi$, $\psi$, $\chi$, \dots, are built on the set $\act$ according to the following grammar:
\begin{align}
		\varphi,\psi & ::=
				{\alpha = \beta} \mid
				\perm(\alpha)
				\mid
				\forb(\alpha)
			\mid
				{\varphi \lor \psi}
			\mid
				{\varphi \land \psi}
			\mid
				{\lnot \varphi}
			\mid
				\bot
			\mid
				\top.
			\label{definition:formulas}
\end{align}%
We use $\form$ to indicate the set of all formulas of \DAL.
Intuitively, action symbols $\mathsf{a} \in \bact$ indicate a \emph{basic} action; actions ${\alpha \sqcup \beta}$ indicate the \emph{free-choice} composition of $\alpha$ and $\beta$; actions ${\alpha \sqcap \beta}$ indicate the \emph{parallel} composition of $\alpha$ and $\beta$; and $\bar{\alpha}$ indicate complement of $\alpha$. Finally, $\mathsf{0}$ and $\mathsf{1}$ indicate the \emph{impossible} and the \emph{universal} actions, respectively.
Turning to formulas, $\alpha = \beta$ indicates that $\alpha$ and $\beta$ are the same actions;  $\perm(\alpha)$ is read as $\alpha$ is \emph{permitted};
and $\forb(\alpha)$ is read as $\alpha$ is \emph{forbidden}.
Formulas built using $\land$, $\lor$, and $\lnot$, as well as $\bot$ and $\top$, have their standard interpretation.
We use $\varphi \to \psi$ as an abbreviation for $\lnot \varphi \lor \psi$, and $\varphi \liff \psi$ as an abbreviation for $(\varphi \to \psi) \land (\psi \to \varphi)$.

\medskip 

In \Cref{ex:syntax}, we present some actions and formulas of \DAL along with their intuitive interpretations.

\medskip

\begin{example}\label{ex:syntax}
Let $\mathsf{parking}$, $\mathsf{drinking}$, and $\mathsf{driving}$ be basic actions in $\bact$. Then:

\medskip
\begin{itemize}
	\item $\overline{\mathsf{parking}} = \mathsf{driving}$ asserts that `parking is the complement of (actively) driving';
	\item $\forb(\mathsf{drinking} \sqcap \mathsf{driving})$ asserting that `drinking while driving is forbidden';
	\item $\forb(\mathsf{drinking} \sqcap \mathsf{driving}) \land \lnot\perm(\mathsf{drinking} \sqcap \mathsf{parking})$ asserting that `drinking while driving is forbidden' and that `it is not permitted to park while driving' either; indicating that operating a vehicle while drinking breaks the law.
\end{itemize}
\end{example}


\paragraph{Semantics.}\label{section:dal:semantics}

The semantics for \DAL is given over deontic action models.
A deontic action model is a tuple $\DeonticModel = \langle E, P, F \rangle$ where: $E$ (the domain) is a (possibly empty) set of elements; and $P$ and $F$ are disjoint subsets of $E$ (i.e., ${P \cup F} \subseteq E$ and ${P \cap F} = \emptyset$).
Intuitively, $E$ indicates realizations of actions, $P$ and $F$ are sets of permitted and forbidden realizations of actions.
The disjointness condition on $P$ and $F$ indicates that permitted realizations of actions are not forbidden, and vice versa, that forbidden realizations are not permitted.
Given a model $\DeonticModel = \tup{E,P,F}$, a \emph{valuation} on $\DeonticModel$ is a function $v: \bact \rightarrow 2^E$.
Intuitively, a valuation indicates a particular way of realizing actions.

\medskip
\begin{proposition}
   For every deontic action model $\DeonticModel = \tup{E,P,F}$, and any valuation $v: {\bact \to 2^E}$ on $\DeonticModel$, there is a unique $v^{*} : {\act \rightarrow 2^E}$ s.t.:
   \begin{align*}
            v^{*}(\alpha \sqcup \beta)
               &=
               {v^{*}(\alpha)
               \cup
               v^{*}(\beta)}
         &
         v^{*}(\bar{\alpha})
               &=
               {E \setminus v^{*}(\alpha)}
         &
         v^{*}(0)
               &=  \emptyset
         \\
            v^{*}(\alpha \sqcap \beta)
               &=
               {v^{*}(\alpha)
               \cap
               v^{*}(\beta)}
         &&&
            v^{*}(1)
               &=  E.
   \end{align*}%
\end{proposition}
\medskip

The \emph{satisfiability} of a formula $\varphi$ on a deontic action model $\DeonticModel = \tup{E, P, F}$ under a
valuation $v$, written ${\DeonticModel, v} \Vdash \varphi$, is
defined inductively as:
\[
\begin{array}{rlcl}
   \DeonticModel, v & \Vdash \alpha=\beta
       & \mathrel{\mbox{ iff }} & 	v^{*}(\alpha) = v^{*}(\beta) \\
   \DeonticModel, v & \Vdash \perm(\alpha)
       & \mathrel{\mbox{ iff }} & v^{*}(\alpha) \subseteq P\\
   \DeonticModel, v & \Vdash \forb(\alpha)
       & \mathrel{\mbox{ iff }} & v^{*}(\alpha) \subseteq F\\
   \DeonticModel, v & \Vdash \varphi \lor \psi
       & \mathrel{\mbox{ iff }} & \DeonticModel, v \Vdash \varphi
         \mbox{ or }      \DeonticModel, v \Vdash \psi\\
   \DeonticModel, v & \Vdash \varphi \land \psi
       & \mathrel{\mbox{ iff }} & \DeonticModel, v \Vdash \varphi
       \mbox{ and } \DeonticModel, v \Vdash \psi\\
   \DeonticModel, v & \Vdash \lnot \varphi
       & \mathrel{\mbox{ iff }} & \DeonticModel, v \nVdash \varphi\\
   \DeonticModel, v & \Vdash \bot
      &  & \mbox{never}\\
   \DeonticModel, v & \Vdash \top
      &  & \mbox{always.}\\
\end{array}
\]
We say that a formula $\varphi$ is a \emph{tautology} iff for any deontic action model $\DeonticModel$ and for any valuation $v$ on $\DeonticModel$, it follows that ${\DeonticModel, v} \Vdash \varphi$.

In \Cref{ex:semantics}, we present examples of deontic action models for the formulas in \Cref{ex:syntax}.

\medskip

\begin{example}\label{ex:semantics}
   \begin{figure}
      \centering
      \begin{minipage}{0.5\textwidth}
            \centering
            \includegraphics[trim=50pt 0pt 50pt 0pt, clip, width=1\textwidth]{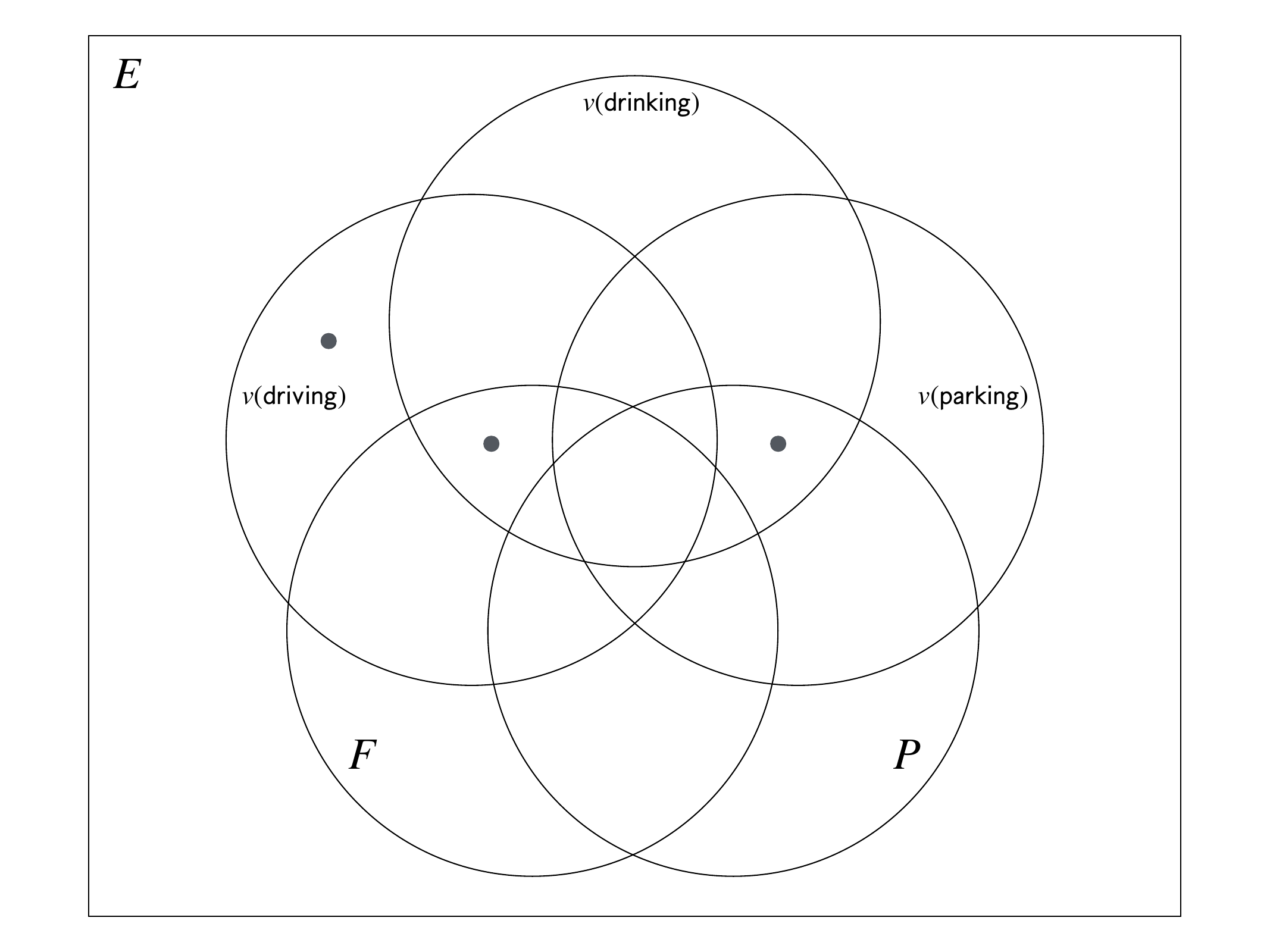}\\[1em] 
            \caption{A Deontic Action Model.}\label{ex:deontic:model:a}
      \end{minipage}\hfill
      \begin{minipage}{0.5\textwidth}
            \centering
            \includegraphics[trim=50pt 0pt 50pt 0pt, clip, width=1\textwidth]{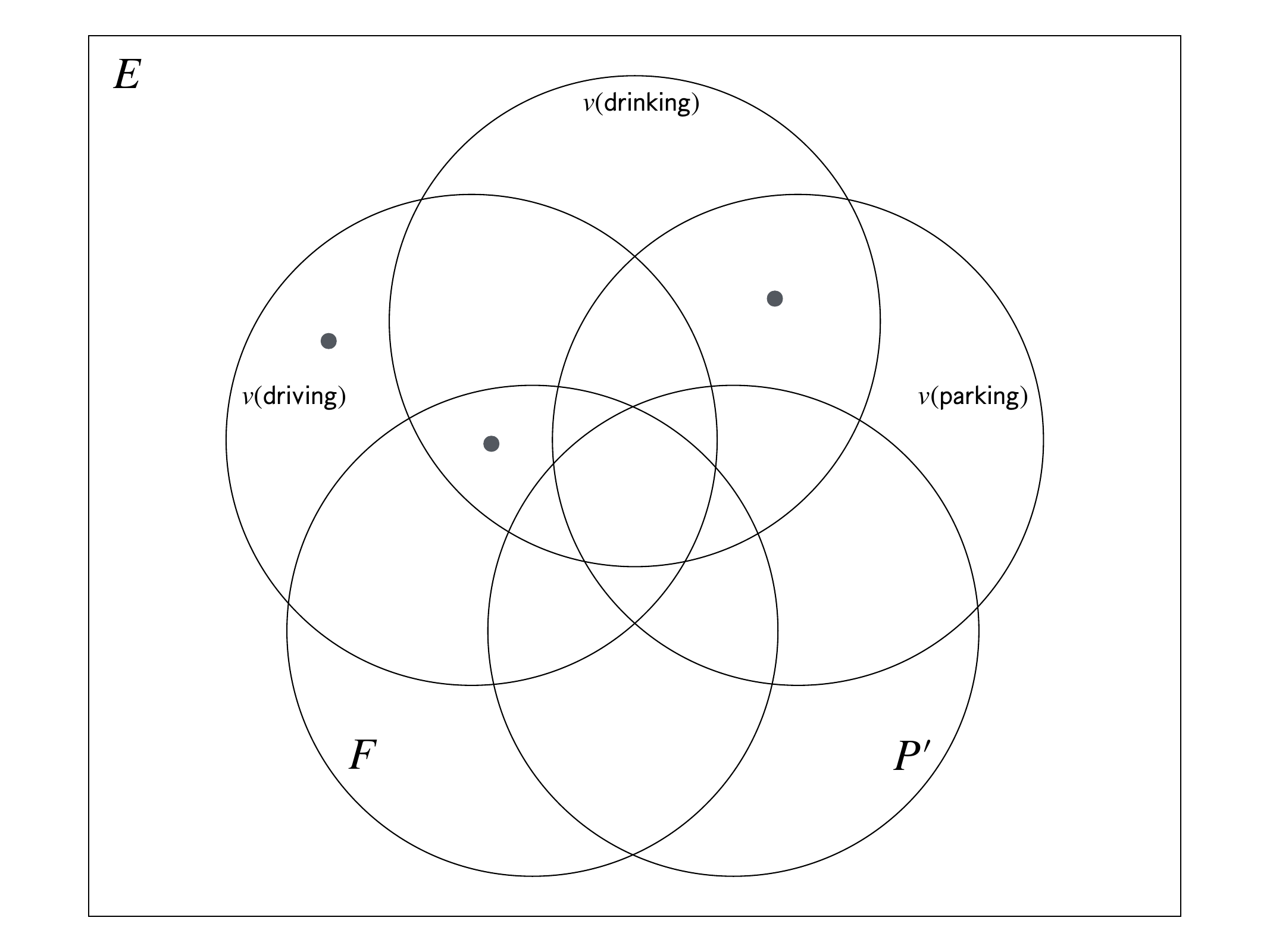}\\[1em] 
            \caption{Anoter Deontic Action Model.}\label{ex:deontic:model:b}
      \end{minipage}
   \end{figure}
   Let $\DeonticModel = \tup{E,P,F}$ be the deontic action model in which $E$, $P$, and $F$ are as in \Cref{ex:deontic:model:a}.
   In addition, $\DeonticModel' = \tup{E,P',F}$ be the deontic action model in which $E$, $P'$, and $F$ are as in \Cref{ex:deontic:model:b}.
   Lastly, let $\{\mathsf{drinking}, \mathsf{driving}, \mathsf{parking}\} \subset \bact$, and $v: \bact \to 2^E$ be a valuation where $v(\mathsf{drinking})$, $v(\mathsf{driving})$, and $v(\mathsf{parking})$ are as in \Cref{ex:deontic:model:a}.
   Then:

   \begin{multicols}{2}
   \begin{enumerate}
      \item $\DeonticModel, v \Vdash \overline{\mathsf{parking}} = \mathsf{driving}$
      \item $\DeonticModel, v \Vdash \forb(\mathsf{drinking} \sqcap \mathsf{driving})$
      \item $\DeonticModel, v \Vdash \perm(\mathsf{drinking} \sqcap \mathsf{parking})$
      \item $\DeonticModel, v \nVdash \forb(\mathsf{drinking} \sqcap \mathsf{driving}) \land$
      \item[] \hspace{1.9cm}$\lnot\perm(\mathsf{drinking} \sqcap \mathsf{parking})$
      \item $\DeonticModel', v \Vdash \overline{\mathsf{parking}} = \mathsf{driving}$
      \item $\DeonticModel', v \Vdash \forb(\mathsf{drinking} \sqcap \mathsf{driving})$
      \item $\DeonticModel', v \nVdash \perm(\mathsf{drinking} \sqcap \mathsf{parking})$
      \item $\DeonticModel', v \Vdash \forb(\mathsf{drinking} \sqcap \mathsf{driving}) \land$
      \item[] \hspace{1.9cm}$\lnot\perm(\mathsf{drinking} \sqcap \mathsf{parking})$.
   \end{enumerate}
   \end{multicols}

   \noindent The models $\DeonticModel$ and $\DeonticModel'$ make clear, for example, that in \DAL, $\forb(\alpha)$ and $\lnot\perm(\alpha)$ are not equivalent.





\end{example}


\paragraph{Axiomatization.}

We present an axiomatic system for \DAL that differs  slightly from the one originally used in \cite{Segerberg1982}. This change is solely motivated by the fact that our presentation simplifies the axiomatization of the variations to \DAL that we introduce in the following sections.

We organize the presentation of the axiom system of \DAL in four groups of axioms as shown in 
\Cref{dal:axioms}. 
The axioms in the first group (A1--A13 and LEM) characterize operations on actions, and is inspired by the presentation of Boolean algebras via complemented distributive lattices in~\cite{Esakia:2019,Halmos:2009}.
The axioms in the second group do the same for propositional connectives on formulas (A1'--A13' and LEM').
%
The axioms in the third group (E1 and E2) characterize equality.
Finally, the axioms the fourth group (D1--D3) characterize the deontic operators of permission $\perm$ and prohibition $\forb$.
%

\begin{figure}
	\centering
	\fbox{
	\begin{minipage}{0.95\textwidth}
		\begin{multicols}{2}
			\begin{enumerate}[label=A\arabic*.]
				\item $\alpha \sqcap (\beta \sqcap \gamma) = (\alpha \sqcap \beta) \sqcap \gamma$
				\item $\alpha \sqcap \beta = \beta \sqcap \alpha$
				\item $\alpha \sqcap \alpha = \alpha$
				\item $\alpha \sqcap (\alpha \sqcup \beta) = \alpha$
				\item ${\alpha \sqcap (\beta \sqcup \gamma)} = {(\alpha \sqcap \beta) \sqcup (\alpha \sqcap \gamma)}$
				\item $\alpha \sqcap \mathsf{0} = \mathsf{0}$
				\item $\alpha \sqcap \bar{\alpha} = \mathsf{0}$
				\item $\alpha \sqcup (\beta \sqcup \gamma) = (\alpha \sqcup \beta) \sqcup \gamma$
				\item $\alpha \sqcup \beta = \beta \sqcup \alpha$
				\item $\alpha \sqcup \alpha = \alpha$
				\item $\alpha \sqcup (\alpha \sqcap \beta) = \alpha$
				\item ${\alpha \sqcup (\beta \sqcap \gamma)} = {(\alpha \sqcup \beta) \sqcap (\alpha \sqcup \gamma)}$
				\item $\alpha \sqcup \mathsf{1} = \mathsf{1}$
				\item[LEM.] $\alpha \sqcup \bar{\alpha} = \mathsf{1}$
			\end{enumerate}
		\end{multicols}
		\ \\[-1.5cm]
		\begin{multicols}{2}
	\begin{enumerate}[label=A\arabic*'.]
		\item $\varphi \wedge (\psi \wedge \chi) \liff (\varphi \wedge \psi) \wedge \chi$
		\item $\varphi \wedge \psi \liff \psi \wedge \varphi$
		\item $\varphi \wedge \varphi \liff \varphi$
		\item $\varphi \wedge (\varphi \vee \psi) \liff \varphi$
		\item ${\varphi \wedge (\psi \vee \chi)} \liff {(\varphi \wedge \psi) \vee (\varphi \wedge \chi)}$
		\item $\varphi \wedge \bot \liff \bot$
		\item $\varphi \wedge \neg \varphi \liff \bot$
		\item $\varphi \vee (\psi \vee \chi) \liff (\varphi \vee \psi) \vee \chi$
		\item $\varphi \vee \psi \liff \psi \vee \varphi$
		\item $\varphi \vee \varphi \liff \varphi$
		\item $\varphi \vee (\varphi \wedge \psi) \liff \varphi$
		\item ${\varphi \vee (\psi \wedge \chi)} \liff {(\varphi \vee \psi) \wedge (\varphi \vee \chi)}$
		\item $\varphi \vee \top \liff \top$
		\item[LEM'.] $\varphi \vee \neg \varphi \liff \top$
	\end{enumerate}
\end{multicols}
		\ \\[-1.5cm]
		\begin{multicols}{2}
			\begin{enumerate}[label=E\arabic*.]
				\item $\alpha = \alpha$
				\item $(\alpha=\beta \land \varphi) \to {\varphi_{\alpha}^{\beta}}$
			\end{enumerate}
		\end{multicols}
		\ \\[-1.5cm]
		\begin{multicols}{2}
			\begin{enumerate}[label=D\arabic*.]
				\item $\perm(\alpha\sqcup\beta) \liff (\perm(\alpha) \land \perm(\beta))$
				\item $\forb(\alpha\sqcup\beta) \liff (\forb(\alpha) \land \forb(\beta))$
				\item $(\perm(\alpha) \land \forb(\alpha)) \liff (\alpha = \mathsf{0})$
			\end{enumerate}
		\end{multicols}
	\end{minipage}}\\[1em]
	\caption{Axiom System for \DAL.}\label{dal:axioms}
\end{figure}





In \DAL, a Hilbert-style proof of a formula $\varphi$ is defined as a finite sequence $\psi_1, \dots, \psi_n$ of formulas s.t.: $\psi_n = \varphi$, and for each $1 \leq k \leq n$, $\psi_k$ is an axiom, or is obtained from two earlier formulas $\psi_i$ and $\psi_j$ using the rule of \emph{modus ponens} (i.e., there are $1 \leq i < j < k$ s.t.\ $\psi_j = {\psi_i \to \psi_k}$).
We say that $\varphi$ is a theorem, and write $\vdash \varphi$, iff there is a proof of $\varphi$.
We make a slight abuse of notation and use \DAL to indicate both the logic and its set of theorems.
We state \Cref{th:segerber:completeness} for future reference.

\medskip
\begin{theorem}[\cite{Segerberg1982}]\label{th:segerber:completeness}
	In \DAL, a formula is a theorem if and only if it is a tautology.
\end{theorem}


\section{Deontic Action Logic via Algebra}\label{sec:algebraic-char}


We now turn our attention to revisiting and expanding the algebraic characterization of \DAL we presented in \cite{CCFA:2021}. To be noted, this algebraic framework is mathematically more abstract compared to the one in \cite{Segerberg1982}. This level of abstraction is a characteristic of algebraic logics, which can be leveraged to address broader issues in deontic logic. Furthermore, a distinguishing feature of our approach is its modularity. The class of algebras described below can be easily extended to support additional deontic operators, and in all cases, standard algebraic tools can be employed to prove soundness and completeness results.
We take advantage of this feature to build new deontic actions logics in the spirit of \DAL in \Cref{section:new:dals}.

\subsection{Basic Definitions (and a Roadmap for our Results)}\label{section:basics}

We provide a brief overview of some fundamental concepts in the algebraization of logic: signatures, algebras, characterizations of classes of algebras, congruences, and quotient algebras. In the case of \DAL, these definitions are interwoven with sorts to distinguish between actions and propositions.
This led us to work with many-sorted algebras --algebraic structures with carrier sets and operations categorized into \emph{sorts}~\cite{Halmos:2009,Tarlecki:2012}. In what follows, we establish the basic terminology for many-sorted algebras in the context of the algebraization of a logic. We have two main purposes behind this: first, to introduce the notation and terminology we use and ensure our results are self-contained; second, to outline the key steps in our algebraization of \DAL.

The algebraization of a logic begins with the appropriate definition of a signature and an algebraic structure. In our case, as mentioned, these two concepts are categorized into sorts.

\medskip
\begin{definition}
	A (many-sorted) signature is a pair $\Sigma = \langle S, \Omega \rangle$ where:
		$S$ is a non-empty set of sort symbols;
		and
		$\Omega$ is an $S^{+}$-indexed family of pairwise-disjoint sets of operation symbols. 
	In turn, an algebra of type $\Sigma$, or a $\Sigma$-algebra, is a structure $\Algebra[A] = \tup{|\Algebra[A]|, \Funcs}$ where:
		$|\Algebra[A]|$ is an $S$-indexed family of non-empty universe sets $|\Algebra[A]|_s$; and
		$\Funcs$ is a collection of functions ${f_{\Algebra[A]}: (\prod_{i = 1}^{n} |\Algebra[A]|_{s_i}) \to |\Algebra[A]|_{s}}$, one for each $f \in \Omega_{s_1 \dots s_n s}$.
\end{definition}
\medskip

For the rest of this section, by an algebra, we mean an algebra of type $\Sigma = \tup{S, \Omega}$.

Signatures give rise to specific algebras whose universe sets consist of strings of symbols from the signature, and whose functions operate as concatenation of these strings. These algebras, known as \emph{term algebras}, serve as the algebraic counterpart to the language of a logic.
We introduce the precise definition of term algebra in \Cref{def:talg}

\medskip
\begin{definition}\label{def:talg}
	Let $V$ be an $S$-indexed family of pair-wise disjoint countable sets of symbols for variables.
	The term algebra $\TAlgebra$ (on $V$) is defined s.t.:
	
	\medskip
	\begin{enumerate}
		\item for all $s \in S$, $|\TAlgebra|_s$ is the smallest set containing:
			$V_s$, and
			all strings $\textrm{`}f(\tau_1 \dots \tau_n)\textrm{'}$ where $f \in \Omega_{s_1 \dots s_n s}$, and $\tau_i \in |\TAlgebra|_{s_i}$;
		\item for all $f \in \Omega_{s_1 \dots s_n s}$,
			$f_{\mathbf{T}}(\tau_1 \dots \tau_n) = \textrm{`}f(\tau_1 \dots \tau_n)\textrm{'}$.
	\end{enumerate}
\end{definition}
\medskip

The algebraic counterpart of the semantics of a logical language is given via homomorphisms of term algebras.

\medskip
\begin{definition}
	Let $\Algebra[A]$ and $\Algebra[B]$ be algebras, and $h = {\set{h_s : {|\Algebra[A]|_s \rightarrow |\Algebra[B]|_s}}{s \in S}}$ be an $S$-indexed family of functions.
	We say that $h$ is a homomor\-phism from $\Algebra[A]$ to $\Algebra[B]$, and write ${h: {\Algebra[A] \to \mathbf{B}}}$, iff for all $f \in \Omega_{s_1 \dots s_n s}$, $h_s(f_{\Algebra[A]}(a_1 \dots a_n)) = f_{\mathbf{B}}(h_{s_1}(a_1) \dots h_{s_n}(a_n))$.
	An interpretation of a term algebra $\TAlgebra$ with variables in $V$ on an algebra $\Algebra[A]$ is a homomorphism $h: \TAlgebra \to \Algebra[A]$.
\end{definition}
\medskip

To establish soundness and completeness results in an algebraic way we will need to connect term algebras to particular classes of algebras of interest.
Standard classes of algebraic structures are characterized by equations. However, our algebraization of \DAL uses the weaker notion of a quasi-equation --a conditional equation-- to capture equality on actions as an algebraic operation. We introduce these concepts next.

\medskip 
\begin{definition}
	Let $\TAlgebra$ be a term algebra with variables in $V$.
	An equation is a string ${\tau_1 \doteq_s \tau_2}$ where $\tau_i \in |\TAlgebra|_s$. 
	In turn, a quasi-equation is a string ${{\tau_1 \doteq_s
	\tau_2} \To {\tau'_1 \doteq_{s'} \tau'_2}}$ where $\tau_i \in |\TAlgebra|_s$ and $\tau'_i \in |\TAlgebra|_{s'}$.
	By ${{\tau_1 \doteq_s
	\tau_2} \Iff {\tau'_1 \doteq_{s'} \tau'_2}}$, we mean the pair of quasi-equations
		${\tau_1 \doteq_s \tau_2} \To {\tau'_1 \doteq_{s'} \tau'_2}$ and
		${\tau'_1 \doteq_{s'} \tau'_2} \To {\tau_1 \doteq_{s} \tau_2}$.
\end{definition}
\medskip 

Notice that equations and quasi-equations are not elements of a term algebra.
Moreover, note that we have used $\doteq$ instead of $=$ in the definition of equations and quasi-equations since $=$, as a symbol, is part of the language of \DAL.
We define below when equations and quasi-equation are satisfied in an algebra.

\medskip
\begin{definition}
	An equation $\tau_1 \doteq_s \tau_2$ is satisfied in an algebra $\Algebra[A]$ under an interpretation $h$ iff ${h_s(\tau_1) = h_s(\tau_2)}$.
	In turn,
		a quasi-equation
			${{\tau_1 \doteq_s	\tau_2} \To {\tau'_1 \doteq_{s'} \tau'_2}}$
		is satisfied in
			$\Algebra[A]$ under $h$
		iff
			${h_s(\tau_1) = h_s(\tau_2)}$ implies ${h_s(\tau'_1) = h_s(\tau'_2)}$.
	Moreover,
		a quasi-equation
			${{\tau_1 \doteq_s	\tau_2} \To {\tau'_1 \doteq_{s'} \tau'_2}}$
		is valid in an algebra
			$\Algebra[A]$
		iff
			for any interpretation $h$ on $\Algebra[A]$,
				${{\tau_1 \doteq_s	\tau_2} \To {\tau'_1 \doteq_{s'} \tau'_2}}$ is satisfied in $\Algebra[A]$ under $h$.
\end{definition}

\medskip

Quasi-equations give rise to classes of algebras called \emph{quasi-varieties}.

\medskip
\begin{definition}
	A quasi-variety is the class of all algebras validating a set of quasi-equations; i.e., the class of all algebras where all quasi-equations in the set is valid.
\end{definition}
\medskip

The final fundamental tool in the algebraic characterization of a logic is that of a congruence relation. When appropriately defined on term algebras, a congruence relation provides a method for constructing --out of syntax-- well-behaved algebras as quotient algebras. More formally, canonical algebraic models are obtained by taking the quotient of the term algebra using specific congruences.

\medskip
\begin{definition}
	Let $\Algebra[A]$ be an algebra, and ${\cong} = \set{{{\cong_s} \subseteq {|\Algebra[A]|_s \times |\Algebra[A]|_s}}}{s \in S}$ be an $S$-indexed family of binary relations.
	We say that $\cong$ is a congruence on $\Algebra[A]$ iff every ${\cong_{s}} \in {\cong}$ is an equivalence relations on $|\Algebra[A]|_{s}$, and for every $f \in \Omega_{s_1 \dots s_n s}$, it follows that
		$a_i \cong_{s_i} a'_i$
		implies
		$f_{\Algebra[A]}(a_1 \dots a_n) \cong_s f_{\Algebra[A]}(a_1' \dots a_n')$.
	The quotient of $\Algebra[A]$ under $\cong$ is an algebra $\Algebra[A]/{\cong}$ where:
		$|\Algebra[A]/{\cong}|_s = |\Algebra[A]|_s/{\cong_s}$;
		and
		$f_{(\Algebra[A]/{\cong})}([a_1]_{\cong_{s_1}} \dots [a_n]_{\cong_{s_n}}) = [f_{\Algebra[A]}(a_1 \dots a_n)]_{\cong_{s}}$.
\end{definition}
\medskip

We conclude this section with a presentation of two classes of algebras we use in our algebraization of \DAL: Boolean and Heyting algebras.
We follow \cite{Esakia:2019} and reach these particular algebras via \emph{bounded distributive lattices} (BDLs).
BDLs are characterized by a set of equations common to Boolean and Heyting algebras.
This gives us a way to introduce concepts pertaining Boolean and Heyting algebras simultaneously.
We introduce Heyting algebras as an extension of BDLs, and  Boolean algebras as an extension of Heyting algebras.
To keep the notation uncluttered, unless it is strictly necessary, we will omit the subscript $\Algebra[A]$ in the function $f_{\Algebra[A]}$ and simply write $f$. The context will always disambiguate whether we refer to the function or the symbol in the signature of $\Algebra[A]$.

\medskip
\begin{definition}[\cite{Esakia:2019}]\label[definition]{def:bdl:algebra}
	Let $\Lambda = \tup{ \{s\}, \{\{{+},{*}\}_{sss},\{0,1\}_{s}\}}$ be a many-sorted signature.
	A BDL-algebra is a $\Lambda$-algebra $\Algebra[L]$ satisfying the following equations:
	\begin{multicols}{2}
		\begin{enumerate}[label=L\arabic*.]
			\item $\tau_1 + (\tau_2 + \tau_3) \doteq (\tau_1 + \tau_2) + \tau_3$
			\item $\tau_1 + \tau_2 \doteq \tau_2 + \tau_1$
			\item $\tau_1 + \tau_1 \doteq \tau_1$
			\item $\tau_1 + (\tau_1 * \tau_2) \doteq \tau_1$
			\item ${\tau_1 + (\tau_2 * \tau_3)} \doteq {(\tau_1 + \tau_2) * (\tau_1 + \tau_3)}$
			\item $\tau_1 + 1 \doteq 1$
			\item $\tau_1 * (\tau_2 * \tau_3) \doteq (\tau_1 * \tau_2) * \tau_3$
			\item $\tau_1 * \tau_2 \doteq \tau_2 * \tau_1$
			\item $\tau_1 * \tau_1 \doteq \tau_1$
			\item $\tau_1 * (\tau_1 + \tau_2) \doteq \tau_1$
			\item ${\tau_1 * (\tau_2 + \tau_3)} \doteq {(\tau_1 * \tau_2) + (\tau_1 * \tau_3)}$
			\item $\tau_1 * 0 \doteq 0$.
		\end{enumerate}
	\end{multicols}
\end{definition}

In \Cref{def:bdl:algebra}, $\tau_i$ is an element of the term algebra of type $\Lambda$.
We use $\Algebra[L] = \tup{L, {+}, {*}, {0}, {1}}$ to indicate an arbitrary BDL-algebra.
We write $\preccurlyeq$ for the partial order implicit in a BDL algebra, i.e., $a \preccurlyeq b$ iff $a + b = b$.

\medskip
Another important concept in our algebraization of \DAL is that of an ideal (or, dually, a filter).
Intuitively, an ideal is an initial set closed by unions (while a filter is a final set closed by intersections).  
Ideals were used by Segerberg in the original definition of \DAL as inherent properties of the formalization of permission and prohibition on actions.

\medskip
\begin{definition}
	An ideal in a BDL-algebra $\Algebra[L] = \tup{L, {+}, {*}, {0}, {1}}$ is a subset $I \subseteq L$ s.t.: for all $x,y \in I$, ${x + y} \in I$,
	and for all $x \in I$ and $y \in L$, $(x*y) \in I$.
	Dual to ideals are filters.
	A filter is a subset $F \subseteq L$ s.t.: for all $x,y \in F$, $(x*y) \in F$, and 
	for all $x \in F$ and $y \in L$, ${x + y} \in F$.
\end{definition}
\medskip

We use BDL-algebras to present Heyting and Boolean algebras. 

\medskip 
\begin{definition}[\cite{Esakia:2019}]\label[definition]{def:heyting:algebra}
	Let $\mathrm{H} = \tup{ \{s\}, \{\{{\hto}, {+},{*}\}_{sss},\{0,1\}_s\}}$ be a many-sorted signature.
	A Heyting algebra is an $\mathrm{H}$-algebra satisfying the equations L1--L12 in \Cref{def:bdl:algebra} together with the following equations:

\medskip
		\begin{enumerate}[label=H\arabic*.]
			\item $\tau_1 * (\tau_1 \hto \tau_2) \doteq \tau_2$
		
			\item $((\tau_1 * \tau_2) \hto \tau_1) * \tau_3 \doteq \tau_3$
			\item $\tau_1 * (\tau_2 \hto \tau_3) \doteq \tau_1 * ((\tau_1 * \tau_2) \hto (\tau_1 * \tau_3))$.
		\end{enumerate}	
\end{definition}
\medskip

In \Cref{def:heyting:algebra}, $\tau_i$ is an element of the term algebra of type $\mathrm{H}$.
We use $\Algebra[H] = \tup{H, {\hto}, {+}, {*}, {0}, {1}}$ to indicate an arbitrary Heyting algebra.
We use $\bar{\tau}$ as an abbreviation of $\tau \hto 0$.
This abbreviation gives rise to an operation $\bar{~} : {H \to H}$ defined as $\bar{x} = {x \hto 0}$.
We will sometimes use $\Algebra[H] = \tup{H, {+}, {*}, \bar{~}, {0}, {1}}$ to indicate an arbitrary Heyting algebra; in these cases, we assume $\hto$ implicitly.

\medskip
\begin{definition}\label{definition:boolean:algebra}
	Boolean algebras are Heyting algebras that validate: (LEM) $\tau + \bar{\tau} \doteq 1$.
\end{definition}
\medskip

We use $\Algebra[B] = \tup{B, {+}, {*}, {\bar{~}}, {0}, {1}}$ to indicate an arbitrary Boolean algebra; and~$\mathbf{2}$ to indicate the Boolean algebra of exactly two elements.
A Boolean algebra is \emph{concrete} iff it is a field of sets.
Important Boolean algebras in our setting are those freely generated and finitely generated~\cite{Halmos:2009}.
We use  Stone's representation theorem~\cite{Stone36}.
If $\Algebra[B]$ is a Boolean algebra, we use $\mathsf{s}(\Algebra[B])$ for its isomorphic Stone space, and ${\mathbf{\varphi}_{\Algebra[B]}: \Algebra[B] \to \mathsf{s}(\Algebra[B])}$ for the isomorphism.


\subsection{Algebraizing Deontic Action Logic}
We start the algebraization of  \DAL introducing its signature, i.e., the symbols needed to capture the language of the logic in an algebraic way.


\medskip
\begin{definition}\label[definition]{def:signature}
The signature of \DAL is a tuple $\Sigma = \tup{S, \Omega}$ where:
		$S = \{\sorta, \sortf\}$; and 
		$\textstyle \bigcup \Omega = \{
			{\sqcup}, {\sqcap}, \bar{~}, \iact, \mathsf{1},
			{\lor}, {\land}, {\lnot}, {\bot}, {\top},
			{=}, {\perm}, {\forb}
		\}$.
	The symbols in $\textstyle \bigcup \Omega$ are further categorized into sets
		$\Omega_{\sorta\sorta\sorta}$,
		$\Omega_{\sorta\sorta}$,
		$\Omega_{\sorta}$,
		$\Omega_{\sortf\sortf\sortf}$,
		$\Omega_{\sortf\sortf}$,
		$\Omega_{\sortf}$,
		$\Omega_{\sorta\sorta\sortf}$,
		$\Omega_{\sorta\sortf}$ summarized in \Cref{tab:sig}.
	
	\begin{figure}
		\centering
		\begin{tabular}{r@{~}lr@{~}lr@{~}lr@{~}lr@{~}l}
			\toprule
			&& \multicolumn{8}{c}{operations}
			\tabularnewline
			\cmidrule{3-10}
			& sorts && actions && formulas && equality && normative
			\tabularnewline
			\midrule
			$S$ & $=\{\sorta, \sortf\}$ &
			$\Omega_{\sorta\sorta\sorta}$ & $= \{ {\sqcup}, {\sqcap}\}$ &
			$\Omega_{\sortf\sortf\sortf}$ & $= \{ {\lor}, {\land}\}$ &
			$\Omega_{\sorta\sorta\sortf}$ & $= \{ {=}\}$ &
			$\Omega_{\sorta\sorta\sortf}$ & $= \{ {\perm}, {\forb}\}$
			\tabularnewline
			&&
			$\Omega_{\sorta\sorta}$ & $= \{\bar{~}\}$ &
			$\Omega_{\sortf\sortf}$ & $= \{{\lnot}\}$ &
			\tabularnewline
			&&
			$\Omega_{\sorta}$ & $= \{\iact, \mathsf{1}\}$ &
			$\Omega_{\sortf}$ & $= \{\bot, \top\}$ &
			\tabularnewline
			\bottomrule
		\end{tabular}\\[1em]
		\caption{The Signature used in the algebraization of \DAL.}\label{tab:sig}
	\end{figure}
\end{definition}
\medskip

In our discussion on the algebraization of \DAL, we take $\Sigma = \tup{S, \Omega}$ to be as in \Cref{def:signature}.
Intuitively, the sort symbols $\sorta$ and $\sortf$ in $S$ categorize actions and formulas, respectively.
In turn, we think of $\Omega$ as containing
symbols for operations on actions, operations on formulas, and
(heterogeneous) operations from actions to formulas.

\medskip
\begin{definition}\label{dal:talg}
	The term algebra $\TAlgebra$ for \DAL uses the set ${\bact}$ as the set of variables of sort $\sorta$, and the empty set $\emptyset$ as the set of variables of sort $\sortf$.  We call this algebra the deontic action term algebra, or the algebraic language of \DAL.
\end{definition}
\medskip

The term algebra $\TAlgebra$ in \Cref{dal:talg} is interpreted over \emph{deontic action algebras}. Deontic action algebras are to \DAL what Boolean algebras are to Classical Propositional Logic, or what Heyting algebras are to Intuitionistic Propositional Logic.
We provide the precise definition of a deontic action algebra in \Cref{definition:deontic:algebra}.

\medskip
\begin{definition}\label[definition]{definition:deontic:algebra}
	A deontic action algebra is an algebra
		$\DAlgebra =
			\langle
				\Algebra[A], \Algebra[F], \E, \P, \F
			\rangle$
	of type $\Sigma$ where:%
		\footnote{
			We use $=$ as the function interpreting `$=$' in $\Algebra[F]$, and $=_{\Algebra[A]}$ and $=_{\Algebra[B]}$ as equality in $\Algebra[A]$ and $\Algebra[B]$, respectively.
		}
		$\Algebra[A] = \tup{A, {\sqcup}, {\sqcap}, \bar{~}, \iact, \uact}$
		and
		$\Algebra[F] = \tup{F, {\lor}, {\land}, {\lnot}, \bot, \top}$ are Boolean algebras,
		and
			$\E$,
			$\P$,
			and
			$\F$,
		satisfy the conditions below
		\begin{multicols}{3}
			\begin{enumerate}[leftmargin=\parindent]
				\item $\P(a {\sqcup} b) \,{=_{\Algebra[F]}}\, {\P(a) {\land} \P(b)}$
				\item $\F(a {\sqcup} b) \,{=_{\Algebra[F]}}\, {\F(a) {\land} \F(b)}$
				\item ${\P(a) {\land} \F(a)} \,{=_{\Algebra[F]}}\, (a \,{=}\, \iact)$
				\item $(a = b) \land \P(a) \preccurlyeq \P(b)$
				\item $(a = b) \land \F(a) \preccurlyeq \F(b)$
				\item[]
				\item ${a \,{=_{\Algebra[A]}}\, b} ~\text{iff}~ {(a \,{=}\, b) \,{=_{\Algebra[F]}}\, \top}$.
			\end{enumerate}
		\end{multicols}
	\noindent
	Let $h: \TAlgebra \to \DAlgebra$ be an interpretation.
	We use $\DAlgebra, h \vDash \tau_1 \doteq \tau_2$ as a shorthand for $h(\tau_1) = h(\tau_2)$.
	In turn, let $\DALVariety$ indicate the class of all deontic action algebras.
	We use $\DALVariety \vDash \tau_1 \doteq \tau_2$ as the universal quantification of $\vDash$ to all deontic action algebras in $\DALVariety$ and all interpretations on these algebras; i.e., $\DALVariety \vDash \tau_1 \doteq \tau_2$ iff $\DAlgebra, h \vDash \tau_1 \doteq \tau_2$, for all $\DAlgebra \in \DALVariety$, and all interpretations $h: \TAlgebra \to \DAlgebra$.
\end{definition}
\medskip

The next two results are immediate.

\medskip
\begin{proposition}\label[proposition]{pro:dal:act2form}
	It follows that $\DALVariety \vDash \alpha \doteq_{\sorta} \beta$ iff $\DALVariety \vDash (\alpha = \beta) \doteq_{\sortf} \top$.
\end{proposition}
\medskip

\begin{proposition}\label[proposition]{pro:dal:qvariety}
	The class $\DALVariety$ of all deontic action algebras is a quasi-variety.
\end{proposition}
\begin{proof}
	It suffices to show that the conditions in the definition of a deontic action algebra can be captured by equations, or quasi-equations.
	The interesting cases are:
	\smallskip
	\begin{enumerate}[leftmargin=\parindent]
		\item $\P(a \sqcup b) =_{\Algebra[F]} {\P(a) \land \P(b)}$ expressed as $\P(a \sqcup b) \doteq_{\sortf} \P(a) \land \P(b)$;
		\item ${\P(a) \land \F(a)} =_{\Algebra[F]} \E(a,\iact)$ expressed as ${\P(a) \land \P(b)} \doteq_{\sortf} {a = \iact}$;
		\item $(a = b) \land \P(a) \preccurlyeq \P(b)$ expressed as $((a = b) \land \P(a)) \lor \P(b) \doteq_{\sortf} \P(b)$; and
		\item ${a =_{\Algebra[A]} b} ~\text{iff}~ {(a = b) =_{\Algebra[F]} \top}$ expressed as the quasi-equations
			${a \doteq_{\sorta} b} \To {(a = b) \doteq_{\sortf} \top}$, and
			${(a = b) \doteq_{\sortf} \top} \To {a \doteq_{\sorta} b}$. \qedhere
	\end{enumerate}
\end{proof}
\medskip

The definition of a deontic action algebra in \Cref{definition:deontic:algebra} draws on ideas and terminology from Pratt's dynamic algebras~\cite{Pratt:1991}. We present the general structure of a deontic action algebra in a form slightly different from the general treatment of many-sorted algebras in \Cref{section:basics}. In doing this we wish to highlight the modular nature of deontic action algebras. In \Cref{section:new:dals}, we leverage this modularity to introduce variants of \DAL by considering different algebraic structures for each component of the deontic action algebra. 
Finally, notice that, as made clear in \Cref{pro:dal:qvariety}, our treatment of equality in the logic results in the class of deontic action algebras being a quasi-variety.
The result in \Cref{pro:dal:act2form} tells us we can dispense explicitly referring to equations on actions as they are also captured as particular equations on formulas via equality in the logic.

\medskip
\begin{example}
	\Cref{ex:deontic:algebra} depicts the deontic action algebra ${\DAlgebra = \langle \Algebra[A], \Algebra[2], \E, \P, \F \rangle}$ where: the algebra $\Algebra[A]$ of actions is the free Boolean algebra on the set of generators $\{a,b\}$.
	In $\DAlgebra$, the functions $\P$ and $\F$ are defined as:
	\begin{align*}
		\P(x) &=
			{\begin{cases}
				1 & \text{if } x \preccurlyeq \bar{b} \\
				0 & \text{otherwise.}
			\end{cases}}
		&
		\F(x) &=
			{\begin{cases}
				1 & \text{if } x \preccurlyeq b \\
				0 & \text{otherwise.}
			\end{cases}}
	\label{eq:ex:pf}
	\end{align*}
	In \Cref{ex:deontic:algebra}, the elements of $|\Algebra[A]|$ that $\P$ and $\F$ map to $\top$ are indicated with green and red, respectively.
	To avoid overcrowding the drawing, we have chosen not to highlight the elements these operations do not map to $\top$.
	In \Cref{ex:deontic:algebra} also, the sets $P$ and $F$ indicate which actions are permitted and which ones are forbidden.
	Note that both sets form an ideal in $\Algebra[A]$ whose intersection contains only the $\iact$ element of the algebra.
	It can easily be inferred from this example that: if $\P(x) = \top$ for all $x \in |\Algebra[A]|$, then, $\F(x) = \bot$ for all $\iact \prec x \in |\Algebra[A]|$.
	Similarly, if $\F(x) = \top$ for all $x \in |\Algebra[A]|$, then, $\P(x) = \bot$ for all $\iact \prec x \in |\Algebra[A]|$.
	These cases are known as \emph{deontic heaven} and \emph{deontic hell}, respectively.
	We will discuss them later on.
\end{example}
\medskip

\begin{figure}
	\centering
	\includegraphics[width=0.5\textwidth]{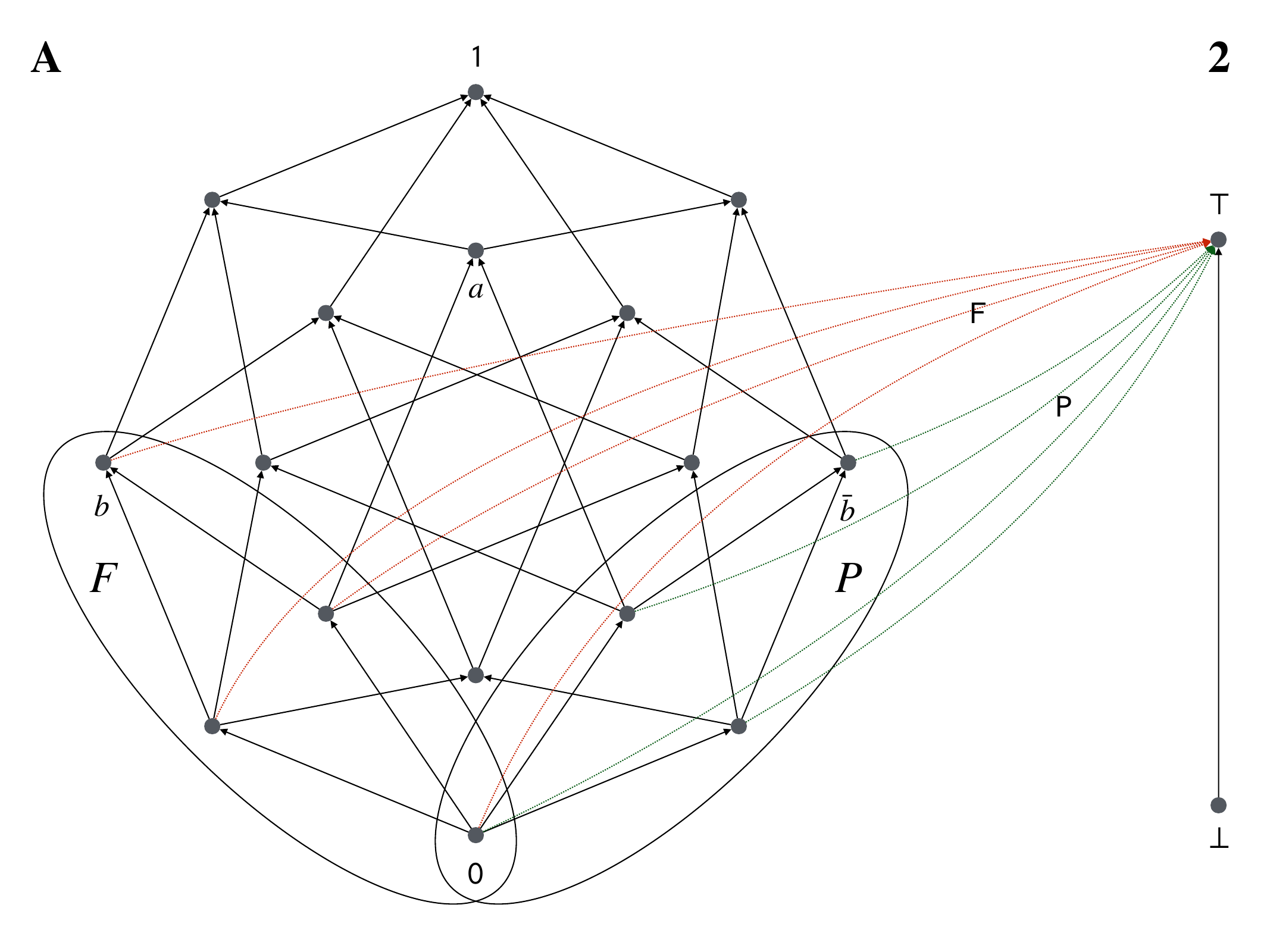}\\[.5em]
	\caption{A Deontic Action Algebra.}\label{ex:deontic:algebra}
\end{figure}

\begin{example}
	Let $\DAlgebra$ be the deontic action algebra in \Cref{ex:deontic:algebra}, and $\mathsf{drinking}$, $\mathsf{driving}$, and $\mathsf{parking}$, be basic action symbols.
	In addition, let $h: \TAlgebra \to \DAlgebra$ be an interpretation s.t.:
		$h_{\sorta}(\mathsf{drinking}) = b$,
		$h_{\sorta}(\mathsf{driving}) = a$, and
		$h_{\sorta}(\mathsf{parking}) = \bar{b}$.
	It follows that:

   \begin{multicols}{2}
   \begin{enumerate}
      \item $h(\overline{\mathsf{parking}} = \mathsf{driving}) =_{\Algebra[F]} \top$
      \item $h(\forb(\mathsf{drinking} \sqcap \mathsf{driving})) =_{\Algebra[F]} \top$
      \item $h(\perm(\mathsf{drinking} \sqcap \mathsf{parking})) =_{\Algebra[F]} \top$
      \item $h(\perm(\mathsf{driving} \sqcup \mathsf{parking})) \neq_{\Algebra[F]} \top$.
   \end{enumerate}
   \end{multicols}

   \noindent In brief, the deontic action algebra $\DAlgebra$ may be understood as the algebraic version of the deontic model $\DeonticModel$ in \Cref{section:dal:semantics}.
\end{example}

The following proposition shows the ideals in the deontic action algebra in \Cref{ex:deontic:algebra} are indeed a distinguishing characteristic of the operations of permission and prohibition.

\medskip
\begin{proposition}\label[proposition]{prop:dal:ideal}
	Let $\DAlgebra = \tup{\Algebra[A], \Algebra[F], \E, \P, \F}$ be a deontic action algebra. The pre-image $P$ of $\top$ under $\P$, as well as the preimage $F$ of $\top$ under $\F$, are ideals in $\Algebra[A]$ s.t.\ ${{P \cap F} = \{\iact\}}$.
\end{proposition}
\begin{proof}
	The result is obtained from the following:
		\medskip
		\begin{enumerate}
			\setlength{\itemsep}{5pt}
			\item
			For all $\{a,b\} \subseteq P$, ${a \sqcup b} \in P$.
			To see why, let $\{a,b\} \subseteq P$.
			Then, $\P(a) = \P(b) = \top$, and $\P(a) \land \P(b) = \top$.
			The properties of $\P$ in \Cref{definition:deontic:algebra} ensure $\P(a) \land \P(b) = \P(a \sqcup b)$.
			This implies $\P(a \sqcup b) = \top$; and so $(a \sqcup b) \in P$.

			\item
			For all $a \in P$ and $b \in |\Algebra[A]|$, ${(a \sqcap b)} \in P$.
			To see why, let $a \in P$ and $b \in |\Algebra[A]|$.
			We know $\P(a) = \top$ and $a = {a \sqcup (a \sqcap b)}$.
			This means $\P({a \sqcup (a \sqcap b)}) = \top$.
			The properties of $\P$ in \Cref{definition:deontic:algebra} ensure $\P(a \sqcup (a \sqcap b)) = \P(a) \land \P(a \sqcap b)$.
			This means, $\P(a) \land \P(a \sqcap b) = \top$.
			From our supposition, $\P(a \sqcap b) = \top$; and so $(a \sqcap b) \in P$.

			\item
			The arguments in 1 and 2 remain true if we replace $P$ and $\P$ for $F$ and $\F$, resp.

			\item
			${P \cap F} = \{ \iact \}$.
			To see why, note that $\P(\iact) = \F(\iact) = \top$; and so $\{\iact\} \subseteq {P \cap F}$.
			In turn, consider an arbitrary $a \in (P \cap F)$.
			Then, $\P(a) = \F(a) = \top$.
			This implies $\P(a) \land \F(a) = \top$, and so $(a = \iact) =_{\Algebra[F]} \top$.
			The `iff' condition in \Cref{definition:deontic:algebra} ensures $a =_{\Algebra[A]} \iact$.
			Since $a$ is arbitrary, the last step tells us that any element in ${P \cap F}$ is equal to $\iact$.
			Therefore, ${P \cap F} \subseteq \{ \iact \}$. \qedhere
		\end{enumerate}
\end{proof}

We proceed to connect the deontic action algebras in $\DALVariety$ with the theorems of $\DAL$.

\medskip
\begin{theorem}[Soundness]\label[theorem]{theorem:soundness}
	If $\varphi$ is a theorem of \DAL, then, $\DALVariety \vDash {\varphi \doteq \top}$.
\end{theorem}
\begin{proof} 
	Let $\DAlgebra \in \DALVariety$ and $h: \TAlgebra \to \DAlgebra$ be any interpretation.
	We continue by induction on the length of the proof of $\varphi$.
	We prove the more interesting cases;  others are similar.

	\medskip
	\begin{enumerate}[leftmargin=\parindent]
		\setlength{\itemsep}{5pt}
		
		\item
		$h_{\sortf}(\perm(\alpha\sqcup\beta) \liff (\perm(\alpha) \land \perm(\beta)))
			=_{\Algebra[F]} \top$.
		The result follows from items (a)--(c) below.

			\medskip
			\begin{enumerate}[leftmargin=\parindent]
				\setlength{\itemsep}{5pt}
				\item
				\begin{description}[leftmargin=\parindent]
					\item[]
					$h_{\sortf}(\perm(\alpha\sqcup\beta) \liff (\perm(\alpha) \land \perm(\beta))) =_{\Algebra[F]}$
					\item[]
					$h_{\sortf}(
						(\lnot \perm(\alpha\sqcup\beta) \lor (\perm(\alpha) \land \perm(\beta)))
						\land
						(\lnot (\perm(\alpha) \land \perm(\beta)) \lor \perm(\alpha\sqcup\beta))
						)=_{\Algebra[F]}$
					\item[]
					$h_{\sortf}(
						\lnot \perm(\alpha\sqcup\beta) \lor (\perm(\alpha) \land \perm(\beta)))
					\land
						h_{\sortf}(
						(\lnot (\perm(\alpha) \land \perm(\beta)) \lor \perm(\alpha\sqcup\beta)))$.
				\end{description}

				\item 
				\begin{description}[leftmargin=\parindent]
					\item[]
					$h_{\sortf}(
						\lnot
							\perm(\alpha\sqcup\beta)
							\lor
							(\perm(\alpha) \land \perm(\beta))) =_{\Algebra[F]}$
					\item[]
						$\lnot
							\perm(h_{\sorta}(\alpha \sqcup \beta))
							\lor
							(\perm(h_{\sorta}(\alpha)) \land \perm(h_{\sorta}(\beta))) =_{\Algebra[F]}$
					\item[]
						$\lnot
							\perm(h_{\sorta}(\alpha \sqcup \beta))
							\lor
							(\perm(h_{\sorta}(\alpha) \sqcup h_{\sorta}(\beta))) =_{\Algebra[F]}$
							\dotfill \Cref{definition:deontic:algebra}(1)
					\item[]
						$\lnot
							\perm(h_{\sorta}(\alpha \sqcup \beta))
							\lor
							(\perm(h_{\sorta}(\alpha \sqcup \beta))) =_{\Algebra[F]} \top$.
				\end{description}

				\item $h_{\sortf}(
					\lnot(\perm(\alpha) \land \perm(\beta))
					\lor
					\perm(\alpha\sqcup\beta)) =_{\Algebra[F]} \top$ is similar to (b).
			\end{enumerate}

		\item
		$h_{\sortf}((\perm(\alpha) \land \forb(\alpha)) \to (\alpha = \iact))
			=_{\Algebra[F]} \top$.
		Then,

			\medskip
			\begin{description}[leftmargin=\parindent]
				\item[]
				$h_{\sortf}((\perm(\alpha) \land \forb(\alpha)) \to (\alpha = \iact)) =_{\Algebra[F]}$
				\item[]
				$h_{\sortf}(
					\lnot (\perm(\alpha) \land \forb(\alpha)) \lor (\alpha = \iact)) =_{\Algebra[F]}$
				\item[]
				$\lnot
					(\perm(h_{\sorta}(\alpha)) \land \forb(h_{\sorta}(\alpha)))
					\lor
					h_{\sortf}(\alpha = \iact) =_{\Algebra[F]}$
				\item[]
				$\lnot
					(\perm(h_{\sorta}(\alpha)) \land \forb(h_{\sorta}(\alpha)))
					\lor
					h_{\sortf}(\alpha = \iact) =_{\Algebra[F]}$
				\item[]
				$\lnot
					(h_{\sorta}(\alpha) = \iact)
					\lor
					h_{\sortf}(\alpha = \iact) =_{\Algebra[F]}$ \dotfill \Cref{definition:deontic:algebra}(3)
				\item[]
				$\lnot
					h_{\sortf}(\alpha = \iact)
					\lor
					h_{\sortf}(\alpha = \iact) =_{\Algebra[F]} \top$.
			\end{description}

		\item
		$h_{\sortf}(((\alpha = \beta) \land \perm(\alpha)) \to \perm(\beta))
			=_{\Algebra[F]} \top$.
			Then,
			
			\medskip
			\begin{description}[leftmargin=\parindent]
				\item[]
					$h_{\sortf}(((\alpha \,{=}\, \beta) {\land} \perm(\alpha)) \to \perm(\beta)) =_{\Algebra[F]}$
				\item[]
					$\lnot
						h_{\sortf}((\alpha \,{=}\, \beta) {\land} \perm(\alpha))
					\lor
					\perm(h_{\sorta}(\beta)) =_{\Algebra[F]}$
				\item[]
					$\lnot
						h_{\sortf}((\alpha \,{=}\, \beta) {\land} \perm(\alpha))
					\lor
						((h_{\sorta}(\alpha) \,{=}\, h_{\sorta}(\beta)) {\land} \perm(h_{\sorta}(\alpha)))
						\lor
						\perm(h_{\sorta}(\beta)) =_{\Algebra[F]}$
						\dotfill \Cref{definition:deontic:algebra}(4)
				\item[]
					$\lnot
						h_{\sortf}((\alpha \,{=}\, \beta) {\land} \perm(\alpha))
					\lor 
						h_{\sortf}((\alpha \,{=}\, \beta) {\land} \perm(\alpha))
					\lor
						\perm(h_{\sorta}(\beta)) =_{\Algebra[F]} \top$.
			\end{description}

			\smallskip
			\noindent
			The result in (3.) is a particular case of the axioms E2 in \Cref{dal:axioms}.
			Other instances can be proven by induction on the size of the formula $\varphi$.
			\qedhere
	\end{enumerate}
\end{proof}

\Cref{theorem:soundness} implies that not every formula of \DAL is provable in the logic.
In particular, non-theorems are not provable.
To see why, consider a theorem $\varphi$, the deontic action algebra $\DAlgebra = \langle \Algebra[A], \Algebra[2], \E, \P, \F \rangle$, and any interpretation $h: \TAlgebra \to \DAlgebra$.
From \Cref{theorem:soundness}, we have $h_{\sortf}(\varphi) = \top$.
Since $h$ is a homomorphism, $h_{\sortf}(\lnot \varphi) = \iact$.
Using the contrapositive of \Cref{theorem:soundness}, $\lnot \varphi$ is not a theorem; i.e., it is not provable.

The converse of \Cref{theorem:soundness}, i.e., the algebraic completeness of \DAL, requires us to show that every non-theorem $\varphi$ of \DAL is falsified in some deontic action algebra $\DAlgebra$ (i.e., there is an interpretation $h: \TAlgebra \to \DAlgebra$ s.t.\ $h_{\sortf}(\varphi) \neq \top$).
We arrive at this result introducing an appropriate notion of congruence, and constructing a quotient algebra via this congruence. 

\medskip
\begin{proposition}\label[proposition]{prop:congruence}
	Let $\TAlgebra$ be the deontic term algebra, and ${\cong_{\sorta}} \subseteq |\TAlgebra|_{\sorta} \times |\TAlgebra|_{\sorta}$ and ${\cong_{\sortf}} \subseteq |\TAlgebra|_{\sortf} \times |\TAlgebra|_{\sortf}$ be s.t.: 1.~$\alpha \cong_{\sorta} \beta$ iff $\alpha = \beta$ is a theorem, and 2.~$\varphi \cong_{\sortf} \psi$ iff $\varphi \liff \psi$ is a theorem.
	It follows that $\cong_{\sorta}$ and $\cong_{\sortf}$ define a congruence $\cong$ on $\TAlgebra$.
\end{proposition}
\medskip

\begin{proposition}\label[proposition]{prop:lindenbaum}
	The quotient of the deontic action term algebra $\TAlgebra$ under $\cong$ is a structure
		$\LTAlgebra = \tup{\Algebra[A], \Algebra[F], \E, \P, \F}$
	where:
		1.~$|\Algebra[A]| = {\act/{\cong_{\sorta}}}$,
		2.~$|\Algebra[F]| = {\form/{\cong_{\sortf}}}$, and
		3.~the operations in $\LTAlgebra$ are those induced by the equivalence classes in $\cong$.
	It follows that $\LTAlgebra \in \DALVariety$.
\end{proposition}
\begin{proof}
	It is clear that $\Algebra[A]$ and $\Algebra[F]$ are Boolean algebras.
	Let us use $\check{\_}$ to indicate the operations in $\LTAlgebra$ induced by $\cong$, and to separate them from the corresponding symbols.
	The result is concluded if $\check{\E}$, $\check{\P}$, and $\check{\F}$ satisfy the conditions in \Cref{definition:deontic:algebra}.
  	We prove some interesting cases only.

	\medskip
	\begin{enumerate}[leftmargin=\parindent]
		\setlength{\itemsep}{7pt}
		
		\item%
		{$\check{\P}([\alpha \sqcup \beta]_{\cong_{\sorta}}) =_{\Algebra[F]} \check{\P}([\alpha]_{\cong_{\sorta}}) \land \check{\P}([\beta]_{\cong_{\sorta}})$}

			\medskip
			\begin{description}[leftmargin=\parindent]
				\item[]
					$\check{\P}([\alpha \sqcup \beta]_{\cong_{\sorta}}) =_{\Algebra[F]}
					[\perm(\alpha \sqcup \beta)]_{\cong_{\sortf}} =_{\Algebra[F]}$
				\item[]
					$[\perm(\alpha) \land \perm(\beta)]_{\cong_{\sortf}} =_{\Algebra[F]}$
					\dotfill \Cref{dal:axioms}(D1)
				\item[]
					$[\perm(\alpha)]_{\cong_{\sortf}} \land [\perm(\beta)]_{\cong_{\sortf}} =_{\Algebra[F]}
					\check{\P}([\alpha]_{\cong_{\sorta}}) \land \check{\P}([\beta]_{\cong_{\sorta}})$.
			\end{description}

		\item%
		{$\check{\P}([\alpha]_{\cong_{\sorta}}) \land \check{\F}([\alpha]_{\cong_{\sorta}}) =_{\Algebra[F]}
			[\alpha]_{\cong_{\sorta}} \mathrel{\check{=}} \iact$}

			\medskip
			\begin{description}[leftmargin=\parindent]
				\item[]
					$
					\check{\P}([\alpha]_{\cong_{\sorta}})
					\land
					\check{\F}([\alpha]_{\cong_{\sorta}})
					=_{\Algebra[F]}
					[\perm(\alpha)]_{\cong_{\sortf}}
					\land
					[\forb(\alpha)]_{\cong_{\sortf}}
					=_{\Algebra[F]}$
				\item[]
					$
					[\perm(\alpha) \land \forb(\alpha)]_{\cong_{\sortf}}
					=_{\Algebra[F]}
					$
				\item[]
					$[\alpha = \iact]_{\cong_{\sortf}} =_{\Algebra[F]}$
					\dotfill \Cref{dal:axioms}(D3)
				\item[]
					$[\alpha]_{\cong_{\sorta}} \mathrel{\check{=}} \iact$.
			\end{description}
		
		\item%
		{$[\alpha]_{\cong_{\sorta}} =_{\Algebra[A]} [\beta]_{\cong_{\sorta}}$ iff
			$[\alpha]_{\cong_{\sorta}} \mathrel{\check{=}} [\beta]_{\cong_{\sorta}} =_{\Algebra[F]} \top$}.

			\medskip
			\begin{description}[leftmargin=\parindent]
				\setlength{\itemsep}{2pt}
				\item[Left-to-right:]
				Let $[\alpha]_{\cong_{\sorta}} =_{\Algebra[A]} [\beta]_{\cong_{\sorta}}$.
				This assumption implies, by definition, that ${\alpha = \beta}$ is a theorem.
				Immediately, ${(\alpha = \beta) \liff \top}$ is also a theorem.
				But this means, $[\alpha = \beta]_{\cong_{\sortf}} =_{\Algebra[F]} \top$.
				Thus, $[\alpha]_{\cong_{\sorta}} \mathrel{\check{=}} [\beta]_{\cong_{\sorta}} =_{\Algebra[F]} \top$.

				\item[Right-to-left:]
				Similarly, let $[\alpha]_{\cong_{\sorta}} \mathrel{\check{=}} [\beta]_{\cong_{\sorta}} =_{\Algebra[F]} \top$.
				Then, $[\alpha = \beta]_{\cong_{\sortf}} = \top$.
				This means $\alpha = \beta$ is a theorem.
				And so, $[\alpha]_{\cong_{\sorta}} =_{\Algebra[A]} [\beta]_{\cong_{\sorta}}$. \qedhere
			\end{description}
	\end{enumerate}	
\end{proof}

We call the quotient algebra $\LTAlgebra$ in \Cref{prop:lindenbaum} the Lindenbaum-Tarski deontic action algebra.
In brief, $\LTAlgebra$ is a canonical algebra that captures theoremhood in \DAL.
From this observation, we obtain the following result.

\medskip
\begin{theorem}[Completeness]\label[theorem]{theorem:completeness}
	If $\DALVariety \vDash {\varphi \doteq \top}$, then, $\varphi$ is a theorem.
\end{theorem}
\begin{proof}
	We show that if $\varphi$ is not a theorem, then $\DALVariety \nvDash {\varphi \doteq \top}$.
	Let $\varphi$ be a non-theorem.
	From the definition of $\cong$, $[\varphi]_{\cong_{\sortf}} \neq_{\Algebra[F]} \top$.
	Define a function $h: {\bact \to {\Algebra[A]/{\cong}}}$ that sends each $\mathsf{a}_i \in \bact$ to the equivalence class $[\mathsf{a}_i]_{\cong_{\sorta}}$.
	The function $h$ extends uniquely to an interpretation $\check{h}: \TAlgebra \to \LTAlgebra$ such that $\check{h}(\varphi) \neq_{\Algebra[F]} [\varphi]_{\cong_{\sortf}}$.
	Therefore, $\DALVariety \nvDash {\varphi \doteq \top}$.
\end{proof}


\begin{corollary}\label[corollary]{cor:completeness}
	$\DALVariety \vDash {\varphi \doteq \top}$ implies $\varphi$ is a tautology.
\end{corollary}
\begin{proof}
	Immediate from \Cref{th:segerber:completeness,theorem:completeness}.
\end{proof}


\subsection{Deontic Action Algebras and Deontic Action Models}

Interestingly, the algebraization of \DAL enjoys a Stone-type representation result connecting the algebraic semantics using deontic action algebras with the original semantics using deontic action models.
This connection provides us with another completeness result for the theorems of $\DAL$.

Recall that Stone's representation theorem establishes that every Boolean algebra is isomorphic to a field of sets~\cite{Stone36}. Such a result reveals a tight connection between the properties of an abstract structure with those of a \emph{concrete} one (a collection of sets).
This is also true for deontic action algebras.
We begin by introducing the definition of a concrete deontic action algebra.


\medskip
\begin{definition}
	A deontic action algebra $\DAlgebra = \langle \Algebra[A], \Algebra[F],  \E, \P, \F \rangle$ is \emph{concrete} iff $\Algebra[A]$ and $\Algebra[F]$ are fields of sets.
  	Let $\DALVariety(0)$ be the class of concrete deontic algebras, for equations of the appropriate sort, we use $\DALVariety(0) \vDash \tau_1 \doteq \tau_2$ as the analogue of $\DALVariety \vDash \tau_1 \doteq \tau_2$ in \Cref{definition:deontic:algebra}.
\end{definition}
\medskip

We prove that validity in deontic action algebras reduces to validity in concrete deontic algebras.
In this way, concrete deontic algebras enable us to connect the algebraic semantics of \DAL with Segerberg's original semantics via Stone's duality.

\medskip
\begin{theorem}\label[theorem]{theorem:reducibility}
	It follows that $\DALVariety(0) \vDash \varphi \doteq \top$ iff $\DALVariety \vDash \varphi \doteq \top$.
\end{theorem}
\begin{proof}
	The left-to-right direction is straightforward.
	The proof for the right-to-left direction is by contradiction.
	Assume that $\DALVariety(0) \vDash \varphi \doteq \top$ and $\DALVariety \nvDash \varphi \doteq \top$.
	This means that we have a deontic action algebra $\DAlgebra = \langle \Algebra[A], \Algebra[F],  \E, \P, \F \rangle$ and an interpretation $h: \TAlgebra \to \DAlgebra$ s.t.\ $h_{\sortf}(\varphi) \neq_{\Algebra[F]} \top$.
	Via the Stone duality result for Boolean algebras, we can construct a concrete deontic action algebra $\DAlgebra' = \langle \Algebra[A]', \Algebra[F]',  \E', \P', \F' \rangle$ that is isomorphic to $\DAlgebra$.
	Moreover, we can define an interpretation $h': \TAlgebra \to \DAlgebra'$ s.t.\ $h'(a_i) = \varphi_{\Algebra[A]'}(h(a_i))$ (with  $\varphi_{\Algebra[A]'}$ being the Stone isomorphism for $\Algebra[A]'$).
	This construction ensures $h'(\varphi) \neq_{\Algebra[F]'} \top$; contradicting the original assumption that $\DALVariety(0) \vDash \varphi \doteq \top$.
\end{proof}

We can now link deontic action models with concrete deontic action algebras.

\medskip
\begin{definition}\label[definition]{def:mod2alg}
	Let $\DeonticModel = \langle E, P, F \rangle$ be a deontic action model, $v:\bact \to 2^{E}$ be a valuation on $\DeonticModel$, and $A = \set{v(\mathsf{a}_i)}{\mathsf{a}_i \in \bact}$.
	Define a concrete deontic action algebra $\algebra(\DeonticModel, v) = \langle \Algebra[A], \Algebra[2], \E, \P, \F \rangle$ where:
	\begin{align*}
			\Algebra[A] &= \tup{2^A, {\cup}, {\cap}, {}^{\complement}, \emptyset, A} &
			(a = b) =_{\Algebra[2]} \top &~\text{iff}~ a =_{\Algebra[A]} b &
			\P(a) = \top &~\text{iff}~ a \subseteq P \\
			&&&& \F(a) = \top &~\text{iff}~ a \subseteq F.
		\end{align*}
	Define also the interpretation $h: \TAlgebra \to \algebra(\DeonticModel, v)$ as the unique extension of $v$.
\end{definition}
\medskip

Similarly, concrete deontic algebras give rise to deontic action models.

\medskip
\begin{definition}\label{def:alg2mod}
	Let $\DAlgebra = \langle \Algebra[A], \Algebra[F],  \E, \P, \F \rangle$ be a concrete deontic algebra, $h: \TAlgebra \to \DAlgebra$ be an interpretation.
	Define a deontic action model $\model(\DAlgebra, h) = \langle E, P, F \rangle$ where:
	\begin{align*}
		E &= |\Algebra[A]| &
		P &= \bigcup \set{a}{\P(a) =_{\Algebra[F]} \top} &
		F &= \bigcup \set{a}{\F(a) =_{\Algebra[F]} \top}.
	\end{align*}
	Define also a valuation $v$ on $\model(\DAlgebra, f)$ as the restriction of $h$ to $\bact$.
\end{definition}
\medskip

If seen as operators, $\model$ and $\algebra$ are inverses of each other.

\medskip
\begin{theorem}\label[theorem]{theorem:inverses}
	It follows that:
		$\algebra(\model(\DAlgebra,v),h)\!=\!\DAlgebra$; and
		$\model(\algebra(\DeonticModel,v),h)\!=\!\DeonticModel$.
\end{theorem}
\medskip

In light of \Cref{theorem:inverses}, we obtain the following result.

\medskip
\begin{corollary}\label[corollary]{theorem:models-and-algebras}
		It follows that:
		$\DeonticModel, v \Vdash \varphi$ iff $\algebra(\DeonticModel, v), h \vDash \varphi \doteq \top$; and
		$\DAlgebra, h \vDash \varphi \doteq \top$ iff $\model(\DAlgebra,h), v \vDash \varphi$.
\end{corollary}
\medskip


\medskip

The results in \Cref{theorem:inverses,theorem:models-and-algebras} enable us to prove the completeness of $\DAL$ w.r.t.\ Segerberg's original deontic models entirely in an algebraic way.

\medskip
\begin{theorem} It follows that $\varphi$ is a theorem iff it is a tautology.
\end{theorem}
\begin{proof}
	Suppose that $\varphi$ is a theorem.
	From \Cref{cor:completeness}, $\DALVariety \vDash \varphi \doteq \top$.
	From \Cref{theorem:reducibility}, $\DALVariety(0) \vDash \varphi \doteq \top$.
	From \Cref{theorem:models-and-algebras}, $\varphi$ is a tautology.
	Thus, $\varphi$ is a theorem implies $\varphi$ is a tautology.
	Using these results in the inverse order we obtain $\varphi$ is a tautology implies $\varphi$ is a theorem.
\end{proof}

\section{Other Deontic Action Logics}\label{section:new:dals}
One of the main benefits of our algebraic treatment of $\DAL$ is that it can be extended using standard algebraic tools to cope with different versions of the logic.  In this section, firstly, we show how classical variations of 
$\DAL$ can be algebraically captured by standard algebraic constructions,  that is, equations, sub-algebras, and generated algebras.  One interesting point of these extensions is that the soundness and completeness properties of these extensions can be obtained by applying similar constructions to the Lindenbaum algebra presented in earlier sections.  Secondly, we consider intuitionistic versions of the logic by replacing the Boolean components of the algebras by Heyting algebras.  As far as we are aware, intutionistic deontic action logics have not been considered before in the literature. 

\subsection{Previously Proposed Variants of \DAL}\label{section:dals}

Segerberg's foundational work~\cite{Segerberg1982} laid the groundwork for a family of closely related deontic action logics. Building on this foundation, the five systems introduced in~\cite{Trypuz15} are particularly interesting as they address specific open issues in the field of Deontic Logic \textemdash such as the \emph{principle of deontic closure}.
We show how our algebraic framework can be easily extended to characterize each of these logics, showcasing the adaptability and versatility of deontic action algebras.
In the rest of this section, we use $\mathsf{a}$ to indicate a basic action symbol and $a$ to indicate its corresponding interpretation in an algebra.

%
%

The first of the five systems in~\cite{Trypuz15}, here called \NDAL{1}, is obtained from \DAL by adding the set $\set{{\forb(\mathsf{a}) \lor \perm(\mathsf{a})}}{\mathsf{a} \in \bact}$ of formulas as additional axioms.
Intuitively, this new set of axioms aims to capture the so-called \emph{principle of deontic closure}\textemdash what is not forbidden is permitted\textemdash at the level of basic actions (i.e., action generators).
The algebraic counterpart of \NDAL{1} is determined by the class of deontic action algebras whose algebra of actions is generated by a set of generators s.t.\ the condition $\F(a) \lor \P(a) =_{\Algebra[F]} \top$ holds for every generator $a$.


The second system, here called \NDAL{2}, is obtained from \NDAL{1} by adding the formula $\perm(\bar{\mathsf{a}}_0 \sqcap \dots \sqcap \bar{\mathsf{a}}_n) \lor \forb(\bar{\mathsf{a}}_0 \sqcap \dots \sqcap \bar{\mathsf{a}}_n)$ as an additional axiom of the logic, under the proviso that $\bact = \{\mathsf{a}_0, \dots, \mathsf{a}_n\}$ for some $n \in \mathbb{N}_0$; i.e., under the proviso that there are finitely many basic action symbols.
Intuitively, this additional axiom states that not performing any of the basic actions is permitted or forbidden.
The algebraic counterpart of \NDAL{2} corresponds to the class of deontic action algebras with a finitely generated atomic Boolean algebra of actions $\Algebra[A]$ satisfying the condition ${
	\P(%
		\bar{a}_1 \sqcap
		\dots \sqcap
		\bar{a}_n
	)
	\lor
	\F(%
		\bar{a}_1 \sqcap
		\dots \sqcap
		\bar{a}_n
	) =_{\Algebra[F]} \top}$
for $\{a_0, \dots, a_n\}$ the set of generators of $\Algebra[A]$.

The third system, \NDAL{3}, is obtained from \NDAL{2} by adding $(\mathsf{a}_0 \sqcup \dots \sqcup \mathsf{a}_n) = \uact$ as an additional axiom.
Intuitively, this new axiom indicates that the agent can only perform actions in $\{\mathsf{a}_1,\dots, \mathsf{a}_n\}$.
The algebraic counterpart of \NDAL{3} corresponds to the subclass of $\NDAL{2}$ further satisfying the condition $a_0 \sqcup \dots \sqcup a_n = \uact$. 

The fourth system, \NDAL{4}, aims to capture the principle of deontic closure at the level of ``atomic'' actions.
Formally, the language of the logic assumes a finite set $\{\mathsf{a}_0, \dots, \mathsf{a}_n\}$ of basic action symbols.
Its axiomatization adds all the formulas in
	$\set
		{\perm({\tilde{\mathsf{a}}_0} \sqcap \dots \sqcap {\tilde{\mathsf{a}}_n}) \lor \forb({\tilde{\mathsf{a}}_0} \sqcap \dots \sqcap {\tilde{\mathsf{a}}_n})}
		{\tilde{\mathsf{a}}_i \in \{\mathsf{a}_i, \bar{\mathsf{a}}_i\}}$ as additional axioms to \DAL.
The algebraic counterpart of \NDAL{4} corresponds to the class of all deontic action algebras with a finitely generated and atomic algebra of actions, whose atoms $a$ satisfy the condition ${\P(a) \lor \F(a) =_{\Algebra[F]} \top}$.

The fifth and last system in \cite{Trypuz15}, here called \NDAL{5}, is the union of \NDAL{3} and \NDAL{4}. Naturally, its
algebraic counterpart corresponds to the intersection of the classes of deontic action algebras characterizing \NDAL{3} and \NDAL{4}.

We now present soundness and completeness results of each of these logics. 
To this end, we introduce the auxiliary definitions of \emph{deontic action subalgebra} and  \emph{deontic action generated algebra}.
Both are analogous to the standard case.

\medskip
\begin{definition}\label{def:deontic-subalgebra} Let $\DAlgebra = \langle \Algebra[A], \Algebra[F], \E, \P, \F \rangle$ and $\DAlgebra' = \langle \Algebra[A]', \Algebra[F]', \E', \P', \F' \rangle$ be two deontic action algebras, we say that 
$\DAlgebra'$ is a subalgebra of $\DAlgebra$ iff: 1.~$\Algebra[A]'$ is a Boolean subalgebra of $\Algebra[A]$; 2.~$\Algebra[F]'$ is a subalgebra of $\Algebra[F]$; and 3.~$\E'$, $\F'$, and $\P'$ are restrictions of $\E$, $\F$, and $\P$ to $\Algebra[A]'$ and $\Algebra[F]'$, respectively.
\end{definition}
\medskip



\begin{definition} Let $\DAlgebra = \langle \Algebra[A], \Algebra[F], \E, \P, \F \rangle$ be a deontic action algebra.
In addition, let $A' \subseteq |\Algebra[A]|$ and $F' \subseteq |\Algebra[F]|$.
The sets $A'$ and $F'$ are called generators.
The deontic action algebra generated by $A'$ and $F'$ is the subalgebra $\DAlgebra = \langle \Algebra[A]', \Algebra[F]', \E', \P', \F' \rangle$ of $\DAlgebra$ where: 1.~$\Algebra[A]'$ is the intersection of all the subalgebras of $\Algebra[A]$ whose carrier set contains $A'$; and 2.~$\Algebra[F]'$ is the intersection of all the subalgebras of $\Algebra[F]$ whose carrier set contains $F'$.
\end{definition}
\medskip

The following theorem extends \Cref{theorem:completeness} for \DAL to its variants \NDAL{i}.

\medskip

\begin{theorem}\label{theorem:completeness:dal:i}
	If follows that $\varphi$ is a theorem of \NDAL{i} iff $\DALVariety(i) \vdash \varphi \doteq \top$.
\end{theorem}
\begin{proof} The proof is direct extension of that in \Cref{theorem:completeness}. We only sketch relevant steps.

\medskip 

\begin{description}
	\setlength{\itemsep}{5pt}
	\item[Soundness.]
	For \NDAL{1} we need to show for all $\DAlgebra \in \DALVariety(1)$ and all interpretations $h: \TAlgebra \to \DAlgebra$, it follows that $h({\forb(\mathsf{a}_i) \lor \perm(\mathsf{a}_i)}) =_{\Algebra[F]} \top$.
	Then:
		\begin{align*}
			h({\forb(\mathsf{a}_i) \lor \perm(\mathsf{a}_i)}) =_{\Algebra[F]}
			h(\forb(\mathsf{a}_i)) \lor h(\perm(\mathsf{a}_i)) =_{\Algebra[F]}
			\F(h(\mathsf{a}_i)) \lor \P(h(\mathsf{a}_i)) =_{\Algebra[F]}
			\top.
		\end{align*}%
	Note that for every $a_i \in \bact$, $h(a_i)$ is a generator, and that homomorphisms between generated Boolean algebras are determined by the mapping between their generators. 
	For the other variants the proofs are similar using the properties of generators, homomorphisms, and the new equations for each case.

	\item[Completeness.]
	Similarly to our result in \Cref{theorem:completeness}, for each \NDAL{i}, we need to define an equivalent to the Lindenbaum-Tarski Algebra of \DAL.  We describe the procedure for \NDAL{1}.
	The other cases use the same argument. 
	First, consider the Lindenbaum-Tarski $\LTAlgebra$ in \Cref{prop:lindenbaum}, and  consider the subalgebra $\LTAlgebra(1)$ generated by the generators 
	$A' =\set{[\mathsf{a}_i]_{\cong_{\sorta}}}{\mathsf{a}_i \in \bact}$, and $F' = \form/_{\cong_{\sortf}}$.
	Furthermore, consider the congruence $\cong_{(1)}$ over $\LTAlgebra(1)$ induced by theoremhood in \NDAL{1}, i.e., the axioms of \DAL plus the new axiom set $\set{\forb(\mathsf{a}_i) \lor \perm(\mathsf{a}_i)}{\mathsf{a}_i \in \bact}$.
	From its construction, $\LTAlgebra(1)/\cong_{(1)}$ is such that	$\F([\mathsf{a}_i]_{\cong_{\sorta}}) \lor \P([\mathsf{a}_i]_{\cong_{\sorta}}) = \top$.
	This algebra provides the canonical deontic action algebra for \NDAL{1}.
	The proof for \NDAL{i}, for $i \in \{2,3,4,5\}$, can be obtained by a similar procedure: a subalgebra of the original Lindenbaum Algebra is considered,  this subalgebra is quotiented by the corresponding axioms, obtaining an algebra that allows us to prove the completenes for the corresponding version of the logic.
	\qedhere
\end{description}
\end{proof}

\subsection{Introducing Propositions}\label{sec:dal:propositions}

Our algebraization of \DAL features an unusual characteristic: the use of an empty set of variables of sort $\sortf$ in the definition of the term algebra $\TAlgebra$ in~\Cref{dal:talg}.
A more natural approach would be to consider a countable set $\prop$ of proposition symbols as variables of sort $\sortf$, analogous to the set $\bact$ of basic action symbols.
Incorporating the set $\prop$ into \DAL results in a new deontic action logic, which we denote as $\DAL(\prop)$.
The construction of this new logic is relatively straightforward, as are its soundness and completeness results.
Moreover, we demonstrate that this new logic has certain advantages over \DAL for modeling scenarios that require explicit propositional reasoning.

\paragraph{Deontic Action Algebras and Propositions.}

The logic $\DAL(\prop)$ extends the logical language of \DAL incorporating symbols in $\prop = \set{p_i}{i \in \Nat_0}$ as base cases in the recursive definition of formulas.
In addition, it adapts the axiom system for \DAL to accommodate for this new definition of a formula.
The algebraization of $\DAL(\prop)$ uses the same signature as \DAL.
Its algebraic language is the term algebra $\TAlgebra_1$ built sets:
	$\bact$ of variables of sort $\sorta$, and
	$\prop$ of variables of sort $\sortf$.
The term algebra $\TAlgebra_1$ is interpreted into deontic action algebras as explained in \Cref{section:basics}. That is, an interpretation of $\TAlgebra_1$ in a deontic action algebra 
$\DAlgebra = \tup{\Algebra[A], \Algebra[F], \E, \P, \F}$ is a homomorphism $h:\TAlgebra_1 \to \DAlgebra$. Note that, by definition, $h$ maps symbols in $\bact$ into elements in $|\Algebra[A]|$ and symbols in $\prop$ into elements in $|\Algebra[F]|$ (in fact, interpretations are completely determined by mappings $\bact \to |\Algebra[A]|$ and $\prop \to |\Algebra[F]|$).

We derive the soundness and completeness of $\DAL(\prop)$ by adapting the corresponding results for \DAL in \Cref{sec:algebraic-char}.

\medskip
\begin{theorem}\label{prop:completeness:dal:prop} $\varphi$ is a theorem of $\DAL(\prop)$ iff $\DALVariety \vDash {\varphi \doteq \top}$.
\end{theorem}
\begin{proof} The proof follows the steps of that of  \Cref{theorem:soundness},  we highlight some subtle details.  

\medskip 
	\begin{description}
		\item[Soundness.]
		The proof of soundness is, as in \Cref{theorem:soundness}, by induction on the length of proofs.
		Note that the set of proofs is defined as in \Cref{theorem:soundness}, but instantiations of axiom schemas may now contain proposition symbols. E.g.,  $(p \land \lnot p) \liff \bot$ is an instance of an axiom schema, and a theorem. This does not affect the proof given in \Cref{theorem:soundness}, which is immediately lifted to a proof of soundness for ${\DAL(\prop)}$.

		\item[Completeness.]
		For completeness, we need to redefine the Lindenbaum algebra.
		First, in this case, we consider the term algebra $\TAlgebra_1$ that contains also formulas with propositions. The congruence $\cong$ in \Cref{prop:lindenbaum} is used to construct the quotient algebra. This quotient algebra is a deontic action algebra. Note that this algebra also contains formula terms with propositions by definition. Finally, adapting the proof of \Cref{theorem:completeness} we obtain the algebraic completeness result.
		\qedhere
	\end{description}
\end{proof}
\medskip

In summary, $\DAL(\prop)$ differs only from $\DAL$ in their respective term algebras.
That is, while \DAL is associated with a term algebra $\TAlgebra$ built over an empty set of variables of sort $\sortf$, the term algebra $\TAlgebra_1$ associated to $\DAL(\prop)$ uses the set $\prop$ as the set of variables of sort $\sortf$.
By adding proposition symbols, the term algebra $\TAlgebra_1$ brings about a sense of correspondence between the basic symbols used for building the set of actions and those used for building the set of formulas.

\paragraph{Propositions Matter}

\DAL, as well as its variants discussed in \Cref{section:dals}, place the focus on formalizing notions of permission and prohibition pertaining to actions.
Nonetheless, they face challenges when presented with statements such as \emph{it is not the case you are permitted to drive without a license}.
This limitation stems from the inability to distinguish between \emph{pure propositions}, such as \emph{you have a driver's license}, and \emph{normative propositions}, such as \emph{you are permitted to drive}.
This distinction is seamlessly addressed in $\DAL(\prop)$.
For example, we can use a proposition symbol $\mathsf{haslicense} \in \prop$ to indicate that a person has a driver's license, and the formula $\perm(\mathsf{driving})$ to indicate that (the action of) driving is permitted.
This allow us to formalize \emph{it is not the case you are permitted to drive without a license} as $\lnot(\lnot\mathsf{haslicense} \land \perm(\mathsf{driving}))$.

Including propositions in deontic action algebras leads to some interesting discussions.
Imagine the following scenario: \emph{there must be no fence;  if there is a fence, then, it must be a white fence; there is a fence}.
This typical case of contrary-to-duty reasoning is discussed in~\cite{Prakken:1996}, where it is noted that prescriptions are applied to propositions rather than actions.
The shift from ought-to-do to ought-to-be is central to most deontic logics developed in the late part of the 20th century; and to \SDL (Standard Deontic Logic) in particular~\cite{Aqvist:2002}.
This shift is not without difficulties. It comes at a cost of quickly leading to paradoxes, i.e., theorems in the logic that are intuitively invalid~\cite{Aqvist:2002,Meyer:1994}.
For instance, in \SDL, a natural formalization of the scenario above, together with the formalization of the global assumption that \emph{if there is a white fence, then, there is a fence} leads to a contradiction; while the scenario is intuitively plausible.

The position held in \cite{Prakken:1996} is that it is worthwhile exploring conditions under which contrary-to-duties can be given consistent readings.
In this respect, we raise the following point.
Up to now, we have assumed that in a deontic action algebra $\DAlgebra = \tup{\Algebra[A], \Algebra[F], \E, \P, \F}$ the algebra $\Algebra[A]$ is used to describe actions.
But a more abstract view of this algebra is also possible.
In particular, we can think of the elements of  $\Algebra[A]$ as entities of the world that can be prescribed, they might be actions, but also propositions such as \emph{there is a fence}, or \emph{the fence is white}.
Under this new reading of a deontic action algebra, in $\DAL(\prop)$, we can distinguish between propositions we can prescribe over, and those we cannot.
For instance, we may regard statements like \emph{it is permitted that is raining} as having little sense, and thus as being ill-formed.
The statement \emph{it is raining} is either true or false, but in no case seems to be amenable to be regulated by a normative system.
Summarizing, deontic action algebras can be used to model ought-to-be normative systems where there is a clear distinction between entities that can be prescribed (corresponding to the elements in $\Algebra[A]$), and those that cannot (corresponding to the elements in $\Algebra[F]$).

\begin{figure}
	\centering
	\includegraphics[width=0.5\textwidth]{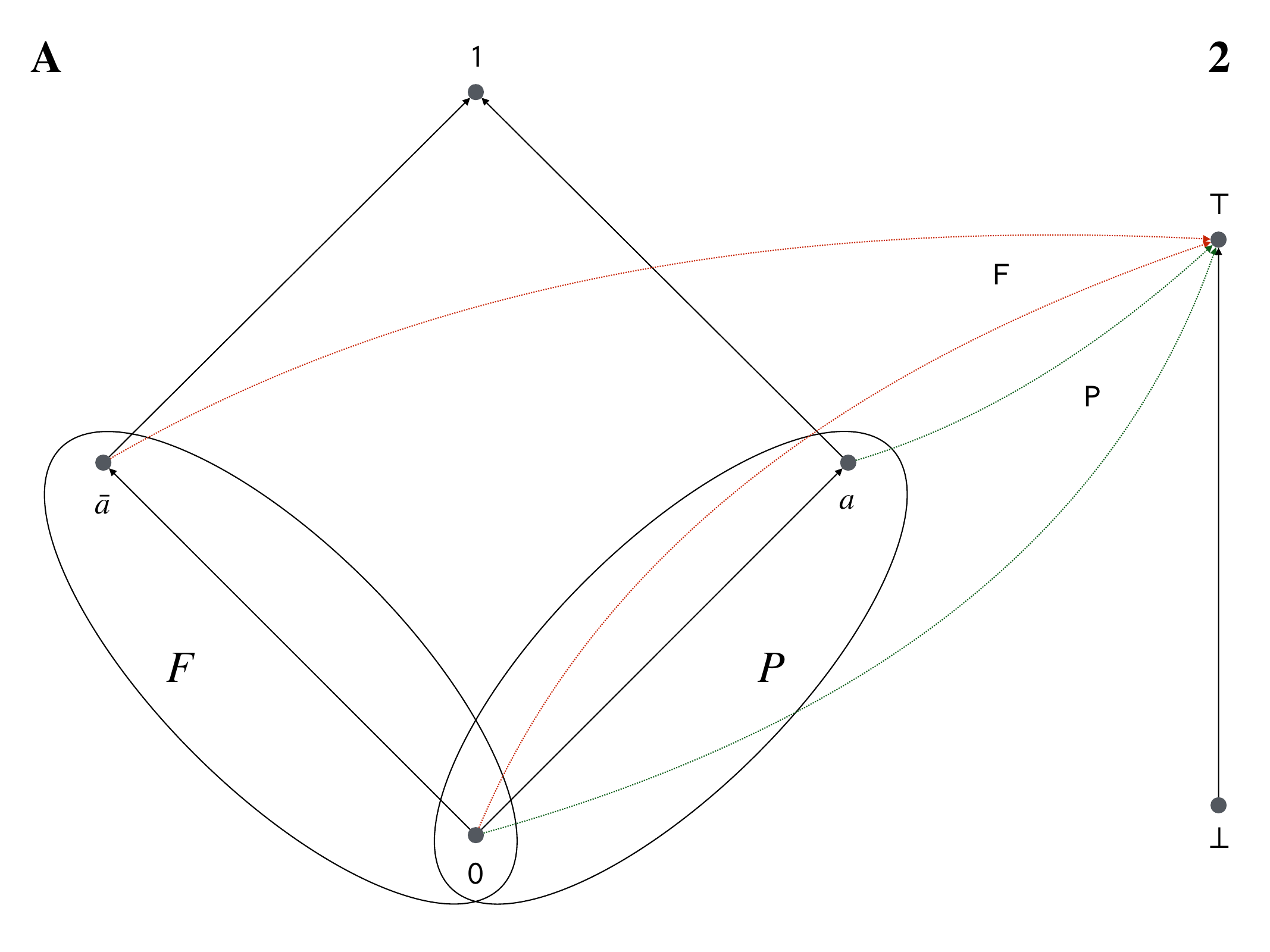}\\[1em]
	\caption{The Cottage Regulations Example}\label{ex:fence}
\end{figure}

\medskip
\begin{example}
	Returning to the example, let us use $\obl(\alpha)$, read as $\alpha$ is \emph{obligatory}, as an abbreviation of $\forb(\bar{\alpha})$.
	Then, we could use the formulas:
	$\obl(\overline{\mathsf{isfenced}})$ 
		to indicate that \emph{there must be no fence},
	$\mathsf{isfenced} = \uact \to \obl(\mathsf{ispaintedwhite})$
		to indicate that \emph{if there is a fence, then, it must be a white fence}, and
	$\mathsf{isfenced} = \uact$
		to indicate that \emph{there is a fence}.
	Finally, we could use the formula
		$\mathsf{ispaintedwhite} \sqcup \mathsf{isfenced} = \mathsf{isfenced}$ to indicate the global assumption that \emph{if the fence is painted white, then, the house is fenced}.
	The deontic action algebra $\DAlgebra$ in \Cref{ex:fence}, together with the interpretation $h: \TAlgebra_1 \to \DAlgebra$ defined as $h_{\sorta}(\mathsf{isfenced}) = \uact$, $h_{\sorta}(\mathsf{ispaintedwhite}) = a$, prove that these formulas are consistent.
	Precisely, we have:

	\medskip
	\centerline{
	\begin{minipage}{0.6\textwidth}
		\begin{enumerate}
			\setlength{\itemsep}{5pt}
			\item $h(\obl(\overline{\mathsf{isfenced}})) =_{\Algebra[F]} \top$.
			\item $h(\mathsf{isfenced} = \uact \to \obl(\mathsf{ispaintedwhite})) =_{\Algebra[F]} \top$.
			\item $h(\mathsf{isfenced} = \uact) =_{\Algebra[F]} \top$.
			\item $h(\mathsf{ispaintedwhite} \sqcup \mathsf{isfenced} = \mathsf{isfenced}) =_{\Algebra[F]} \top$.
		\end{enumerate}
	\end{minipage}}
\end{example}
\medskip

To sum up, we have explored some key features and applications of the incorporation of propositions into deontic action algebras. These insights, brought about by our discussion and examples, are particularly relevant to understand some broader implications of Segerberg's formalization of the notions of permission and prohibition.
\subsection{Heyting Algebras for Formulas}\label{sec:heyting:formulas}

Let us now turn to leveraging the modular framework of deontic action algebras in the construction of new deontic action logics.
In \Cref{sec:algebraic-char}, we brought attention to this modularity presenting a deontic action algebra as a structure $\DAlgebra = \tup{\Algebra[A], \Algebra[F], \E, \P, \F}$, with $\Algebra[A]$ and $\Algebra[F]$ interpreting actions and formulas, and $\E$, $\P$, and $\F$ formalizing equality, permission, and prohibition of actions.
While we have primarily considered $\Algebra[A]$ and $\Algebra[F]$ as Boolean algebras, our framework allows also for alternative algebras for actions and formulas.
Notably, defining $\Algebra[F]$ as a Heyting algebra leads to a new deontic action logic worth considering.
We call this new logic $\DAL(\IPL)$.
We begin with an outline of the technical foundations of $\DAL(\IPL)$, and follow with a discussion of its key features and advantages.

\paragraph{Constructive Reasoning in Deontic Action Algebras}

The language of $\DAL(\IPL)$ contains the actions and formulas of $\DAL(\prop)$. Namely, actions are built using basic action symbols in $\bact$, and the connectives $\sqcup$, $\sqcap$, $\bar{~}$, $\iact$, and $\uact$.
In turn, formulas are built using proposition symbols in $\prop$, the deontic connectives on actions, i.e., $\perm(\alpha)$ and $\forb(\alpha)$, and the connectives $\lor$, $\land$, $\lnot$, $\bot$, and $\bot$.
The sole difference is that $\DAL(\IPL)$ introduces the connective $\to$ as primitive rather than as an abbreviation \textemdash with $\varphi \liff \psi$ remaining as an abbreviation for $(\varphi \to \psi) \land (\psi \to \varphi)$.
The axiomatization of $\DAL(\IPL)$ uses the axioms in \Cref{dal:axioms} for actions, equality, and the deontic operations, while the axioms for the propositional connectives (A1'--A13' and LEM') are replaced by those in \Cref{axioms:ipl}.
These last axioms are standard for Intuitionistic Propositional Logic~\cite{Troelstra:1988}.
Provability and theoremhood are straightforwardly adapted to accommodate for the new axioms.

\begin{figure}
	\centering
	\fbox{
	\begin{minipage}{1.0\textwidth}
		\setlength{\linewidth}{.97\textwidth}
		\setlength{\columnsep}{-1.6cm}
		\begin{multicols}{2}
			\begin{enumerate}[label=H\arabic*.]
				\item $\varphi \to (\varphi \lor \psi)$
				\item $\varphi \to (\psi \lor \varphi)$
				\item $\varphi \land \psi \to \varphi$
				\item $\varphi \land \psi \to \psi$
				\item $(\varphi \to \bot) \to \lnot \varphi$
				\item $\lnot \varphi \to (\varphi \to \bot)$
				\item $\bot \to \varphi$
				\item $\varphi \to \top$
				\item $\varphi \to ( \psi \to \varphi )$
				\item $\varphi \to (\psi \to (\varphi \land \psi))$
				\item ${(\varphi \to \chi) \to ((\psi \to \chi) \to ((\varphi \lor \psi) \to \chi))}$
				\item ${(\varphi \to (\psi \to \chi)) \to ((\varphi \to \psi) \to (\varphi \to \chi))}$
			\end{enumerate}
		\end{multicols}
	\end{minipage}}\\[1em]
	\caption{Axiom System of $\DAL(\IPL)$}\label{axioms:ipl}
\end{figure}

The algebraization of $\DAL(\IPL)$ replaces the Boolean algebra of formulas in the definition of a deontic action algebra for a Heyting algebra. The precise definition of the new type of deontic action algebra being used is given below.

\medskip
\begin{definition}\label[definition]{def:dalgebra:heyting:boolean}
	A BH-deontic-action algebra is an algebra
		$\DAlgebra =
			\langle
				\Algebra[A], \Algebra[H], \E, \P, \F
			\rangle$
		where:
			$\Algebra[A]$ is a Boolean algebra,
			$\Algebra[H]$ is a Heyting algebra, and
			$\E : {|\Algebra[A]| \times |\Algebra[A]| \to |\Algebra[H]|}$,
			$\P : {|\Algebra[A]| \to |\Algebra[H]|}$,
			and
			$\F : {|\Algebra[A]| \to |\Algebra[H]|}$ satisfy the conditions 1--6 in \Cref{definition:deontic:algebra}.
\end{definition}
\medskip

In brief, Heyting algebras play a role in constructive reasoning analogous to the role Boolean algebras play in classical reasoning. A key distinction is how Heyting algebras treat $\to$.
Despite this difference, Heyting algebras are closely related to Boolean algebras.
Specifically, every Boolean algebra is a Heyting algebra, and the regular elements of a Heyting algebra \textemdash those $x 
\in |\Algebra[H]|$ for which $x =_{\Algebra[F]} \lnot\lnot x$ \textemdash form a Boolean algebra.
It is well-known also that Heyting algebras have a representation theorem \textemdash as the category of Heyting algebras is dually equivalent to the category of Eusaki spaces.
Furthermore, the Lindenbaum algebra obtained from the axioms in \Cref{axioms:ipl} is itself a Heyting algebra \cite{vanDalen:2008}.
These facts collectively support the idea of replacing Boolean algebras with Heyting algebras in the algebraic treatment of deontic action logic, ensuring that such an approach is well-founded.

The proposition below exposes an interesting feature of BH-deontic-action algebras.

\medskip
\begin{proposition}\label[proposition]{prop:ideals-int}
	Let $\DAlgebra = \tup{\Algebra[A], \Algebra[H], \E, \P, \F}$ be a BH-deontic-action algebra.
	The pre-images $P$ and $F$ of $\top$ under $\P$ and $\F$, respectively, are ideals in $\Algebra[A]$ s.t.\ ${{P \cap F} = \{\iact\}}$.
\end{proposition}
\begin{proof}
	Note that ideals in Heyting algebras and ideals in Boolean algebras are defined identically.
	Note also that the proof of the analogous result for deontic action algebras in \Cref{prop:dal:ideal} uses only reasoning on ideals and the properties of $\E$, $\P$, and $\F$.
	Since these properties are maintained in BH-deontic-action algebras, the proof in \Cref{prop:dal:ideal} transfers directly to this new setting.
\end{proof}
\medskip

In line with \Cref{prop:dal:ideal}, the result in \Cref{prop:ideals-int} tells us that permission and prohibition on actions yielding ideals carry over if we replace Boolean for Heyting algebras.

By way of conclusion, we present soundness and completeness theorems for $\DAL(\IPL)$.
Definitions of interpretations of the term algebra into BH-deontic-action algebras, homomorphisms, and congruences and quotients, are akin to those in \Cref{sec:algebraic-char}.

\medskip
\begin{theorem}\label{prop:completeness:heyting:formulas}
	 Let $\mathbb{BH}$ be the class of all BH-deontic-action algebras.
	 It follows that $\varphi$ is a theorem of $\DAL(\IPL)$ iff $\mathbb{BH} \vDash {\varphi \doteq \top}$.
\end{theorem}
\begin{proof}
	Like with the proofs of \Cref{theorem:completeness:dal:i,prop:completeness:dal:prop}, we only remark on the differences with the proofs of \Cref{theorem:soundness,theorem:completeness}.
	\begin{description}
		\item[Soundness.]
		We need to prove that any interpretation of an axiom is mapped to $\top$.
		For axioms on actions, this is just like in \Cref{theorem:soundness}.
		For axioms on propositional connectives, this is  immediate from well known results for
		Heyting algebras (see~\cite{vanDalen:2008,Troelstra:1988}).
		Finally, the cases of equality, permission, and prohibition are not affected by the new interpretation of $\to$ in a Heyting algebra.

		\item[Completeness.]
		We begin by defining the Lindenbaum algebra as in \Cref{prop:lindenbaum} via congruences $\cong_{\sorta}$ for actions and $\cong_{\sortf}$ for formulas.
		Again, following~\cite{vanDalen:2008}, it is easy to see that the axioms in \Cref{axioms:ipl} result in the algebra of formulas itself being a Heyting algebra.
		This construction provides a witness for theoremhood for the logic.
		\qedhere
	\end{description}
\end{proof}

\paragraph{Constructive Reasoning Matters}

Just like we did when we introduced propositions, let us discuss why interpreting formulas on Heyting algebras instead of Boolean algebras bears an interest beyond its formal properties.

To set the stage for discussion, imagine the following scenario: \emph{if John does not have a driver's license, then, it is forbidden for him to drive; John is not forbidden to drive}.
From this scenario, using classical reasoning, we derive \emph{John has a driver's license}.
This conclusion is somewhat counterintuitive.
It is easy to consider many cases in which \emph{John does not have a driver's license} is true, which are consistent with the scenario in question.
But this is logically impossible.
There are many ways of dealing with this kind of problem, one of which is to move from Classical to Intuitionistic reasoning. 

\medskip
\begin{example}
	The driver's license example in the previous paragraph can be formalized with formulas:
	$\lnot \mathsf{haslicense} \to \forb(\mathsf{driving})$
		capturing that \emph{if John does not have a driver's license, then, it is forbidden for him to drive}, and
	$\lnot\forb(\mathsf{driving})$
		capturing that \emph{John is not forbidden to drive}.
	Note that, in this formalization, $\mathsf{haslicense} \in \prop$ and $\mathsf{driving} \in \bact$.
	The BH-deontic-action algebra $\DAlgebra$ in \Cref{ex:driver}, together with the interpretation $h$ defined as $h_{\sorta}(\mathsf{driving}) = a$, $h_{\sortf}(\mathsf{haslicense}) = \frac{1}{2}$, prove that $\mathsf{haslicense}$ is not a consequence of the previous two formulas.
	Precisely, we have:

	\medskip
	\centerline{
	\begin{minipage}{0.7\textwidth}
		\begin{enumerate}[leftmargin=\parindent]
			\setlength{\itemsep}{5pt}
			\item $
				h(\lnot\mathsf{haslicense}) =_{\Algebra[H]}
				\lnot h(\mathsf{haslicense}) =_{\Algebra[H]}
				\lnot \frac{1}{2} =_{\Algebra[H]}
				\bot$.
			\item $
				h(\lnot\mathsf{haslicense} \to \forb(\mathsf{driving})) =_{\Algebra[H]}
				\bot \to h(\forb(\mathsf{driving})) =_{\Algebra[H]}
				\top$.
			\item $
				h(\lnot\forb(\mathsf{driving})) =_{\Algebra[H]}
				\lnot h(\forb(\mathsf{driving})) =_{\Algebra[H]}
				\lnot \bot =_{\Algebra[H]}
				\top$.
			\item $h(\mathsf{haslicense}) \neq_{\Algebra[H]} \top$.
		\end{enumerate}
	\end{minipage}}
	\medskip

	\noindent Note how in the BH-deontic-algebra $\DAlgebra$ in \Cref{ex:driver} the only element of the algebra of actions which the operation $\F$ maps to $\top$ is $\iact$, all other elements are mapped to $\bot$.
\end{example}
\medskip


\begin{figure}
	\centering
	\begin{minipage}{0.48\textwidth}
		\centering
		\includegraphics[trim=20pt 0pt 20pt 0pt, clip, width=\textwidth]{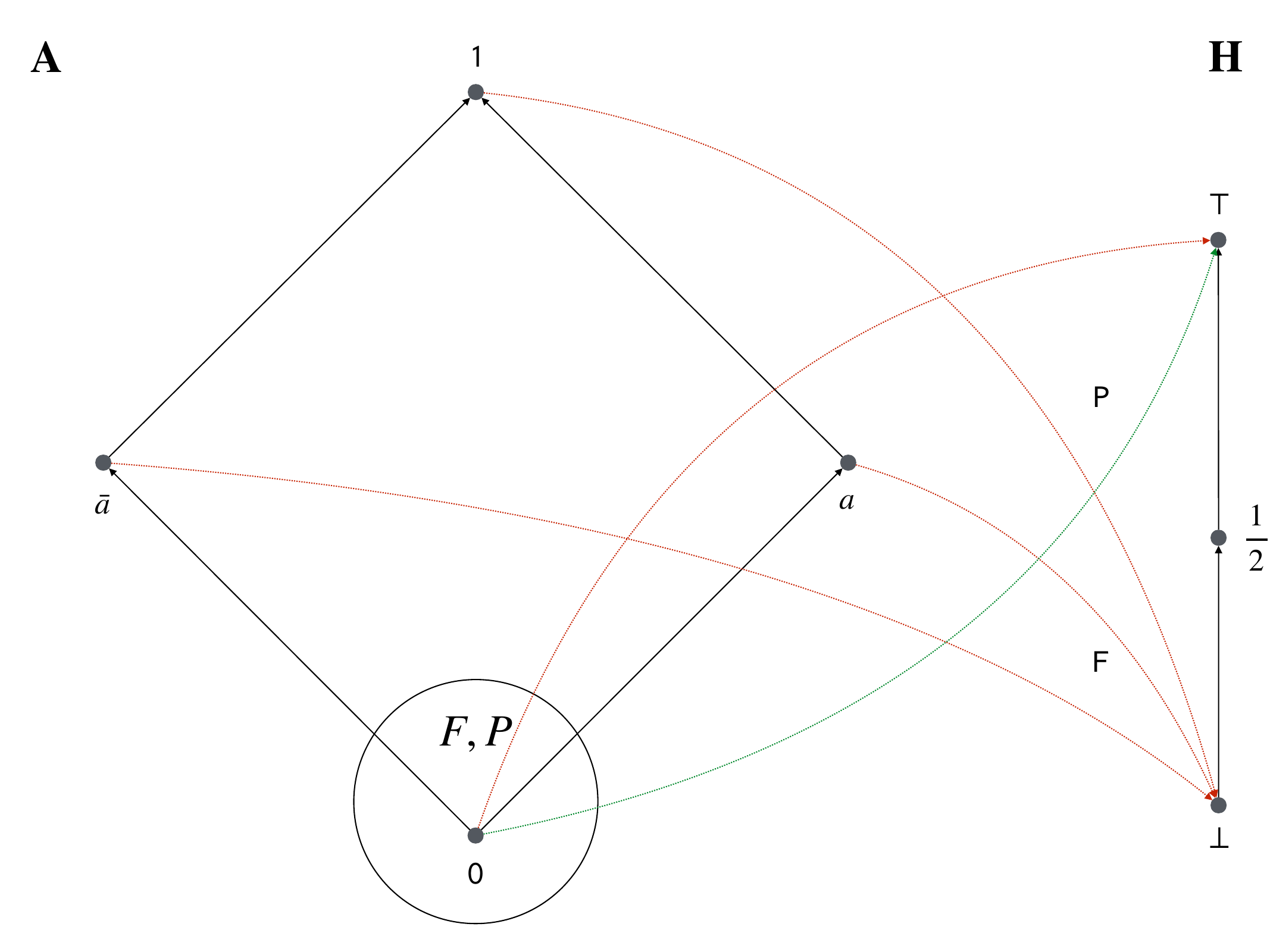}\\[1em]
		\caption{The Driver's License Paradox.}\label{ex:driver}
	\end{minipage}\hfill
	\begin{minipage}{0.48\textwidth}
		\centering
		\includegraphics[trim=20pt 0pt 20pt 0pt, clip, width=\textwidth]{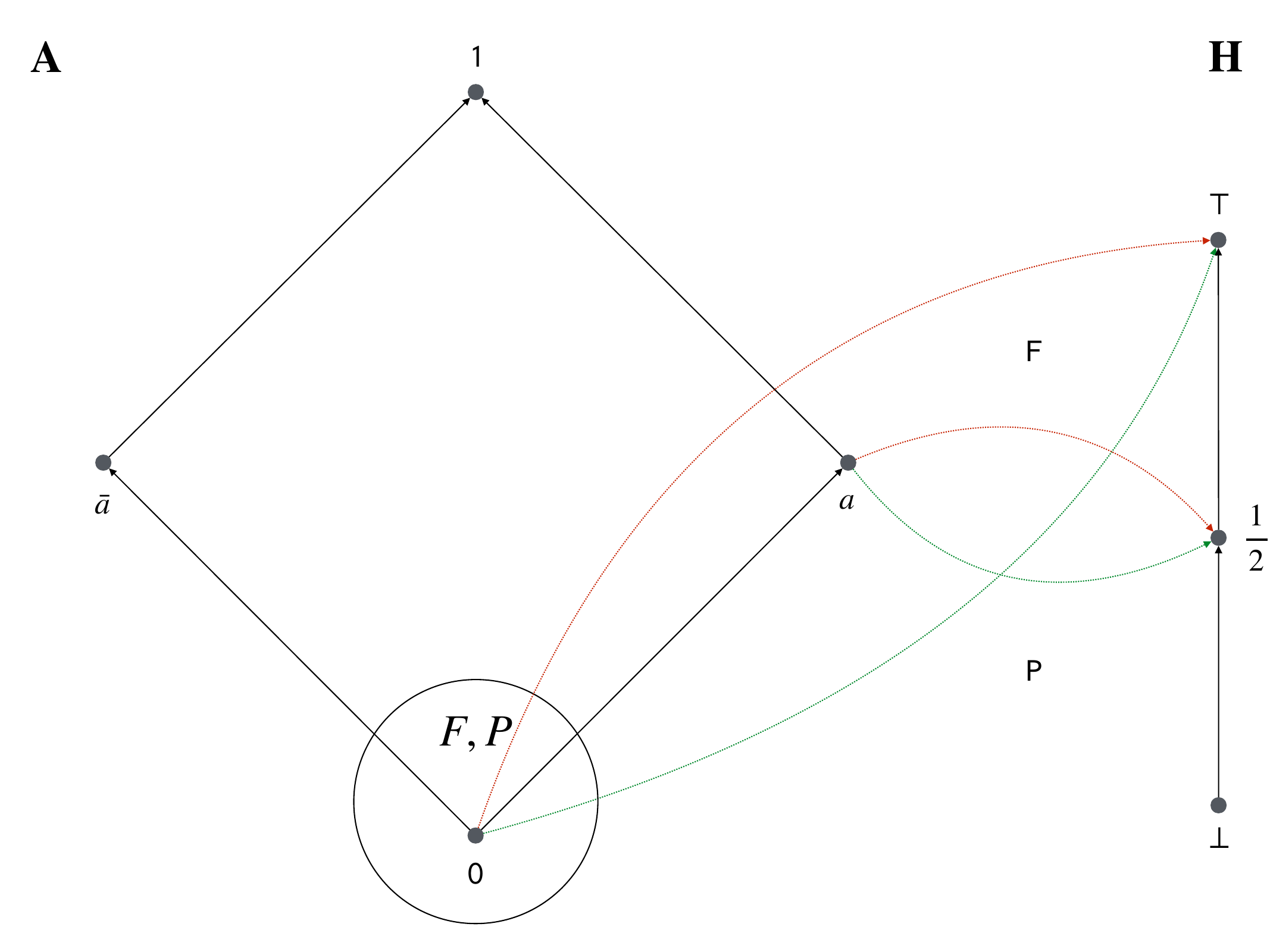}\\[1em] 
		\caption{Principle of Deontic Closure.}\label{ex:deontic:closure}
	\end{minipage}
\end{figure}

There is another interesting discussion emerging from the use of an Intuitionistic basis for reasoning about formulas.
Recall the deontic action logic $\DAL(1)$ from \Cref{section:dals} and how this logic is built from $\DAL$ adding additional axioms with the intent to capture the \emph{principle of deontic closure}.
This principle is stated in \cite{Segerberg1982} as: \emph{what is not forbidden is permitted}.
The formalization of this principle as formulas of the form $\forb(\mathsf{a}) \lor \perm(\mathsf{a})$ is taken from \cite{Trypuz15}.
Still, a more faithful formalization of this principle is $\lnot\forb(\mathsf{a}) \to \perm(\mathsf{a})$.
Clearly, there is no substantial distinction in a Classical setting, as both formulas are equivalent.
This is not the case in an Intuitionistic setting.
For instance, the BH-deontic-action algebra $\DAlgebra$ in \Cref{ex:deontic:closure} satisfies one version of the principle but not the other.
Precisely, note that if we have a single basic action symbol $\mathsf{a} \in \bact$, and an interpretation $h$ on $\DAlgebra$ s.t.\ $h_{\sorta}(\mathsf{a}) = a$, then:

	\medskip
	\centerline{
	\begin{minipage}{0.5\textwidth}
		\begin{enumerate}[leftmargin=\parindent]
			\setlength{\itemsep}{5pt}
			\item $
				h(\forb(\mathsf{a})) =_{\Algebra[H]}
				\F(h(\mathsf{a})) =_{\Algebra[H]}
				\F(a) =_{\Algebra[H]}
				\frac{1}{2}$.
			\item $
				h(\lnot\forb(\mathsf{a})) =_{\Algebra[H]}
				\lnot h(\forb(\mathsf{a})) =_{\Algebra[H]}
				\lnot \frac{1}{2} =_{\Algebra[H]}
				\bot$.
			\item $
				h(\perm(\mathsf{a})) =_{\Algebra[H]}
				\P(h(\mathsf{a})) =_{\Algebra[H]}
				\P(a) =_{\Algebra[H]}
				\frac{1}{2}$.
			\item $
				h(\lnot\forb(\mathsf{a}) \to \perm(\mathsf{a})) =_{\Algebra[H]}
				\bot \to h(\perm(\mathsf{a})) =_{\Algebra[H]}
				\top$.
			\item $
				h(\forb(\mathsf{a}) \lor \perm(\mathsf{a})) =_{\Algebra[H]}
				\frac{1}{2} \lor \frac{1}{2} =_{\Algebra[H]}
				\frac{1}{2} \neq_{\Algebra[H]}
				\top$.
		\end{enumerate}
	\end{minipage}}
	\medskip

\noindent In words, this example shows that there is a distinction between considering the principle of deontic closure as \emph{what is not forbidden is permitted} \textemdash alternatively, \emph{what is not permitted is forbidden}\textemdash and considering this principle as \emph{every (basic) action is either permitted or forbidden}.

To sum up, we have explored some key features and applications of replacing the Boolean algebra of formulas in a deontic action algebra for a Heyting algebra.
The results we obtained underscore leveraging the modularity of our framework to build a new deontic action logic $\DAL(\IPL)$. The discussion and ensuing examples reinforce the utility of this new logic in the broader area of Deontic Logic, and in particular in relation to the principle of deontic closure.

\subsection{A Heyting Algebra of Actions}\label{sec:action-int}

The formal machinery in \Cref{sec:heyting:formulas} naturally suggests its symmetric extension: replacing the Boolean algebra of actions with a Heyting algebra.
This results in a new deontic action logic, $\DAL(\IAL)$, where actions are interpreted in a manner analogous to formulas in an intuitionistic framework.
To the best of our knowledge, no existing deontic action logic provides an intuitionistic perspective on actions, making this approach a novel contribution to the field.

\paragraph{Another look at Constructive Reasoning in Deontic Action Algebras}

We begin with an outline of the technical foundations of $\DAL(\IAL)$.
The formulas of this new logic are built using proposition symbols in $\prop$, the deontic connectives on actions, i.e., $\perm(\alpha)$ and $\forb(\alpha)$, and the connectives $\lor$, $\land$, $\lnot$, $\bot$, and $\bot$.
In turn, actions are built using basic action symbols in $\bact$, and the connectives $\sqcup$, $\sqcap$, $\bar{~}$, $\iact$, and $\uact$.
In addition, $\DAL(\IAL)$ introduces a new connective $\hto$ on actions giving rise to actions of the form $\alpha \hto \beta$.
This new connective is introduced to capture the notion of a relative complement (or intuitionistic implication) in a Heyting algebra.
The axiomatization of $\DAL(\IAL)$ uses all the axioms in \Cref{dal:axioms} except the axiom (LEM) for actions.
In addition, it introduces as axioms the properties H1\textendash H3 in \Cref{def:heyting:algebra} for the new connective $\hto$.
In essence, the axioms for actions are the conditions on Heyting algebras in~\cite{Esakia19}.
Provability and theoremhood are easily adapted to accommodate for the new axioms.

The algebraization of $\DAL(\IAL)$ replaces the Boolean algebra of actions in the definition of a deontic action algebra for a Heyting algebra. This is made precise in \Cref{def:dalgebra:heyting:actions} below.

\medskip
\begin{definition}\label[definition]{def:dalgebra:heyting:actions}
	An HB-deontic-action algebra is an algebra
		$\DAlgebra =
			\langle
				\Algebra[H], \Algebra[F], \E, \P, \F
			\rangle$
		where:
			$\Algebra[H]$ is a Heyting algebra,
			$\Algebra[F]$ is a Boolean algebra, and
			$\E : {|\Algebra[H]| \times |\Algebra[H]| \to |\Algebra[B]|}$,
			$\P : {|\Algebra[H]| \to |\Algebra[B]|}$,
			and
			$\F : {|\Algebra[H]| \to |\Algebra[B]|}$ satisfy the conditions 1--6 in \Cref{definition:deontic:algebra}.
\end{definition}
\medskip

\Cref{prop:heyting:actions:ideal} shows that permission and prohibition behave as expected.

\medskip
\begin{proposition}\label[proposition]{prop:heyting:actions:ideal}
	Let $\DAlgebra = \tup{\Algebra[H], \Algebra[F], \E, \P, \F}$ be an HB-deontic-action algebra.
	The pre-images $P$ and $F$ of $\top$ under $\P$ and $\F$, respectively, are ideals in $\Algebra[A]$ s.t.\ ${{P \cap F} = \{\iact\}}$.
\end{proposition}
\begin{proof}
	Analogous to that in \Cref{prop:ideals-int}.
\end{proof}
\medskip


Since $\DAL(\IAL)$ is the symmetric counterpart of $\DAL(\IPL)$, the proofs of soundness and completeness for $\DAL(\IAL)$ can be straightforwardly adapted from those of $\DAL(\IPL)$. Consequently, we establish the following theorem.

\medskip
\begin{theorem}\label{prop:completeness:heyting:actions}
	 Let $\mathbb{HB}$ be the class of all HB-deontic-action algebras.
	 It follows that $\varphi$ is a theorem of $\DAL(\IAL)$ iff $\mathbb{HB} \vDash {\varphi \doteq \top}$.
\end{theorem}

\paragraph{Constructive Reasoning and Realization of Actions}

We put forth the argument that an intuitionistic basis for actions is useful when actions are tied to constructions that witness their realizability.
This parallels the standard interpretation of Intuitionistic Logic, where the truth of a formula corresponds to the existence of a proof.
Interpreting actions on an intuitionistic basis is not only of theoretical interest but also hold potential for practical applications, particularly in automated planning~\cite{GNT:2016}.
For example, consider a robot capable of executing various actions.
To perform an action, the robot requires plans\textemdash sequences of basic activities that realize the action.
In such a scenario, we may reject $\bar{a} \sqcup a = \uact$, as the robot might lack a plan to execute action $a$ or a way to determine if $a$ is unrealizable.

Prescriptions often play an important role in planning.
For instance, if the robot is an autonomous vehicle, it must adhere to transit rules.
In this case, $\perm(\alpha)$ indicates that plans for executing $\alpha$ are permitted, while $\forb(\alpha)$ signals that such plans are forbidden.
This perspective highlights how an intuitionistic basis for interpreting actions aligns well with practical considerations in scenarios where realizability and prescriptive constraints on actions are central.
In this respect, $\DAL(\IAL)$ provides a logical framework that is well-suited to addressing real-life issues.

\subsection{Intuitionistic Deontic Action Logic}

Clearly, we can also simultaneously replace the Boolean algebras for actions and formulas for Heyting algebras.
We call the resulting logic $\DAL(\INT)$.
Similarly to the case in \Cref{sec:action-int}, to the best of our knowledge, this logic is the first fully intuitionistic deontic action logic.


\paragraph{The Logic Itself}

The language and axiomatization of $\DAL(\INT)$ combines those of $\DAL(\IAL)$ and $\DAL(\IPL)$ in \Cref{sec:heyting:formulas,sec:action-int}.
Precisely, it builds actions like in $\DAL(\IAL)$\textemdash using $\hto$ as a primitive connective.
In turn, it builds formulas like in $\DAL(\IPL)$\textemdash using $\to$ as a primitive connective.
The axiomatization of this new logic takes the axioms for actions from $\DAL(\IAL)$ and the axioms for formulas from $\DAL(\IPL)$.
The logic retains the axiomatization of equality, permission, and prohibition of Segerberg's logic, i.e., axioms E1--E2 and D1--D3 in \Cref{dal:axioms}.
The notions of proof and theoremhood are reformulated in the obvious way.

The algebraization of $\DAL(\INT)$ replaces the Boolean algebras of actions and of formulas in the definition of a deontic action algebra for Heyting algebras. This is made precise in the definition of an \emph{intuitionistic deontic action algebra} below.

\medskip
\begin{definition}\label[definition]{def:intuitionistic:dalgebra}
	By an intuitionistic deontic action algebra, we mean an algebra
		$\DAlgebra =
			\langle
				\Algebra[A], \Algebra[H], \E, \P, \F
			\rangle$
		where:
			$\Algebra[A]$ and 
			$\Algebra[F]$ are Heyting algebras, and
			$\E : {|\Algebra[A]| \times |\Algebra[A]| \to |\Algebra[F]|}$,
			$\P : {|\Algebra[A]| \to |\Algebra[F]|}$,
			and
			$\F : {|\Algebra[A]| \to |\Algebra[F]|}$ satisfy the conditions 1--6 in \Cref{definition:deontic:algebra}.
\end{definition}
\medskip

As before, permission and prohibition behave as expected.

\medskip
\begin{proposition}\label[proposition]{prop:intuitionistic:ideal}
	Let $\DAlgebra = \tup{\Algebra[A], \Algebra[F], \E, \P, \F}$ be an intuitionistic deontic action algebra.
	The pre-images $P$ and $F$ of $\top$ under $\P$ and $\F$ are ideals in $\Algebra[A]$ s.t.\ ${{P \cap F} = \{\iact\}}$.
\end{proposition}
\medskip
\begin{proof}
	We can reuse the proof of \Cref{prop:ideals-int} as it only uses the properties of $\P$ and $\F$ plus absorption, and idempotence properties which also hold in Heyting algebras.
\end{proof}

For stating and proving the soundness and completeness of $\DAL(\INT)$, we define interpretations and algebraic validity as in previous sections.

\medskip

\begin{theorem}\label{prop:completeness:heyting}
	Let $\mathbb{ID}$ be the class of all intuitionistic deontic action algebras.
	Then, $\varphi$ is a theorem of $\DAL(\INT)$ iff $\mathbb{ID} \vDash {\varphi \doteq \top}$.
\end{theorem}
\begin{proof}
	We obtain this result by putting together the intuitionistic parts of the proofs of
	\Cref{prop:completeness:heyting:formulas,prop:completeness:heyting:actions}.
%
%
\end{proof}

\paragraph{Intuitionistic Deontic Action Logic in Practice}

The logic $\DAL(\INT)$ may be useful for reasoning in scenarios where there is partial observability about the state of affairs, typical in reinforcement learning and planning.
Consider the following example adapted from \cite{DBLP:conf/iros/CassandraKK96}: a robot is tasked with cleaning an office and needs to reach certain spots.
The robot has sensors to detect doorways, walls, or open spaces, but the sensor information may sometimes be unclear.
Proposition symbols like $\mathsf{north}$, $\mathsf{south}$, $\mathsf{east}$, and $\mathsf{west}$ could represent the robot's orientation, while $\mathsf{doorway}$, $\mathsf{wall}$, and $\mathsf{clear}$ capture the information provided by the sensors.
Uncertainty entails the robot might fail to determine whether there is a doorway ahead or not, i.e., $\mathsf{doorway} \lor \lnot\mathsf{doorway}$ may fail to hold, violating the law of the excluded middle at the level of formulas.
The example can be extended with an intuitionistic model of actions.
For instance, we can consider actions $\mathsf{advance}$ and $\mathsf{rotate}$ for the robot moving forward and rotating, respectively.
As in the example in \Cref{sec:action-int}, the interpretation of these actions is tied to a plan allowing the robot to realize the action.
Once again, a formula $\mathsf{advance} \sqcup \overline{\mathsf{advance}} = \uact$ may fail to hold (violating the law of excluded middle at the level of actions) because the robot lacks a plan due to incomplete sensor information.

Finally, we can spice up this scenario with prescriptions.
For instance, there might be signals in the corridors indicating the robot is not to cross certain doorway, prohibiting such an action.
Again, the robot may lack sufficient information to determine whether $\forb(\mathsf{advance}) \lor \lnot\forb(\mathsf{advance})$ holds or not.

In all the above scenarios, having an intuitionistic deontic action logic like $\DAL(\INT)$ provides a formal framework for analysis in which we can address issues arising from partial observability and prescriptive constraints.

\section{Final Remarks}\label{section:conclusion}

We developed an algebraic framework for Deontic Action Logic (\DAL) and its variations using deontic action algebras.
These structures consist of two Boolean algebras connected by operations that capture key aspects of permission and prohibition.
We showed that the algebraic characterization is
adequate by proving soundness and completeness theorems.
We discussed the advantages of our algebraic approach for modelling scenarios through concrete examples.
Our algebraic treatment of \DAL can be thought of as an abstract version of deontic action logics which can
be used to establish connections between deontic action logics and areas such as topology, category theory, probability, etc.
Moreover, the framework is modular.
In \Cref{section:new:dals}, we showed how replacing the underlying algebraic structures allows for the development of new logics, highlighting the flexibility and extensibility of our approach.

We introduced deontic action algebras in~\cite{CCFA:2021}. In this article, we extended our previous work and explored the inclusion of Heyting algebras to obtain intuitionistic behavior.
Our approach accommodates such a formulation in a very simple manner, paving the way for interesting future work.
Several alternative algebraic structures for actions and propositions warrant further investigation.
In particular, we aim to characterize action composition (denoted by $;$), and action iteration (denoted by $*$).
These operations are not foreign in deontic reasoning.
The work in~\cite{Meyer:88} on Dynamic Deontic Logic ($\mathsf{DDL}$) was one of the first in considering a deontic logic containing action composition.
This treatment, however, is not without challenges~\cite{Anglberger:08}.
Regarding action iteration, in~\cite{BroersenThesis}, Broersen pointed out that dynamic deontic logics can be divided into: (i) \emph{goal norms}, where prescriptions over a sequence of actions only take into account the last action performed; or (ii) \emph{process norms}, where a sequence of actions is permitted/forbidden if and only if every action in the sequence is permitted/forbidden.
To the best of our knowledge, no extensions of \DAL incorporating action composition or iteration have been explored.
The framework of deontic action algebras is well-suited for exploring such extensions, as deontic action algebras can be straightforwardly modified to admit action composition and iteration.
More precisely, we may consider deontic action algebras $\langle \Algebra[A], \Algebra[F], \P, \F, \E \rangle$ where the algebra $\Algebra[A] = \tup{A, +, ;, *}$ of actions is a Kleene algebra (see~\cite{Kozen:90}).
Kleene algebras enjoy some nice properties, e.g.,
they are quasi-varieties, and they are complete w.r.t.\ equality of regular expressions (see~\cite{Kozen:91}).
Similarly, one can extend deontic action algebras with other interesting algebras, e.g., relation algebras (see~\cite{MadduxBook}) that most notably provide action converse.
We leave it as further work to investigate the properties of the
operators $\P$ and $\F$ in these new algebraic settings.

Beyond Boolean and Heyting algebras, it is also interesting to explore alternative algebras for propositions.
Some immediate examples include: BDL-algebras, semi-lattices, and metric spaces.
This may lead to the design of deontic logics that are not logics of normative propositions but logics of norms instead---a distinction was already noted by von~Wright in~\cite{vonWright:1999} and Alchourr\'on in~\cite{Alchourron:69,Alchourron:1971}.
Both \SDL and \DAL are logics of normative propositions insofar as they assign truth values to formulas in the logic.
In contrast, logics of norms can express prescriptions that do not carry with them truth values.
To accommodate such logics, we can generalize the interpretation of formulas to other algebraic structures.
For instance, adopting a meet semi-lattice as the algebra of propositions allows for norms to be combined while also accounting for potential contradictions among them without necessarily requiring norms to be true or false.
Of course, several other algebraic frameworks could serve this purpose as well.

The level of flexibility in \Cref{section:new:dals} suggests a possible connection between \DAL and combining logics~\cite{sep-logic-combining}.
Building on the algebraic treatment of \DAL, we can derive a characterization in category-theoretic terms~\cite{MacLane98}, which naturally connects to the framework of institutions introduced by Goguen and Burstall~\cite{GoguenBurstall84,GoguenBurstall92}.
Institutions provide an abstract framework for model theory that captures the essence of logical systems and their combinations, offering a unified perspective on how logics can be integrated through categorial constructions.
Many of these methods are unified within the algebraic fibring approach introduced by Sernadas, Sernadas, and Caleiro~\cite{Sernadas1999}, which significantly enhances the versatility of logic combination through universal categorial constructions.
This approach extends the range of logics that can be combined beyond modal logics, demonstrating a fruitful interplay between the algebraic and categorial perspectives.
Altogether, these frameworks establish a strong link between \DAL and broader methodologies for combining logics, reinforcing the relevance of algebraic and categorial approaches to logical systems.

Finally, the algebraization of DAL in this article provides a rich foundation for studying deontic action logics dynamically, \emph{\'a la} to Public Announcement Logic~\cite{Plaza2007}.
In particular, we note that the algebraic semantics of deontic operators induces a restriction on the algebra of formulas.
Such a restriction resembles the model update operators in dynamic logics, suggesting an interesting parallel and potential applications in modeling evolving normative systems. Exploring these connections, alongside the broader algebraic and categorial perspectives outlined above, offers a promising direction for future research.

\paragraph{Conflict of Interest}

The authors declare that there are no conflicts of interest regarding the publication of this paper.

\paragraph{Data Availability Statement}

No new data were created or analyzed in this paper.





\bibliography{bibliography}


\end{document}